\documentclass[journal]{IEEEtran}

\usepackage{cite}
\usepackage{amsmath,amsfonts,amssymb}
\usepackage{graphicx}
\usepackage{epstopdf}
\usepackage{bbm}

\usepackage{dsfont}
\usepackage{amsthm}
\usepackage{float}
\usepackage{subfigure}
\usepackage{stfloats}
\usepackage{bm}
\usepackage{siunitx}
\usepackage{algorithm}
\usepackage{algpseudocode}

\usepackage{balance}

\graphicspath{{Figures/}}

\usepackage{xcolor}
\definecolor{darkgreen}{RGB}{0,100,0}
\definecolor{darkblue}{RGB}{0,0,139}
\definecolor{darkred}{rgb}{0.55, 0.0, 0.0}

\usepackage[colorlinks=true, urlcolor = darkblue, citecolor = darkgreen, linkcolor=darkred]{hyperref}

\newtheorem{lemma}{\textbf{Lemma}}

\newtheorem{remark}{\textbf{Remark}}

\newcommand{\order}[1]{\ensuremath{{#1}\textsuperscript{th}}}

\usepackage{mathtools}

\makeatletter
\def\ps@IEEEtitlepagestyle{%
  \def\@oddfoot{\mycopyrightnotice}%
  \def\@oddhead{\hbox{}\@IEEEheaderstyle\leftmark\hfil\thepage}\relax
  \def\@evenhead{\@IEEEheaderstyle\thepage\hfil\leftmark\hbox{}}\relax
  \def\@evenfoot{}%
}
\def\mycopyrightnotice{%
  \begin{minipage}{\textwidth}
  \centering \scriptsize
  \copyright~2026 IEEE. Personal use of this material is permitted.  Permission from IEEE must be obtained for all other uses, in any current or future media, including reprinting/republishing this material for advertising or promotional purposes, creating new collective works, for resale or redistribution to servers or lists, or reuse of any copyrighted component of this work in other works.
  \end{minipage}
}
\makeatother

\begin{document}
\bstctlcite{IEEEexample:BSTcontrol}

\title{ Physical-Layer Anonymous Communication Design Under Symbol-Aware Identification and Transmitter-Dependent Noises}

\author{Yu~Li, Milad~Tatar~Mamaghani, Xiangyun~Zhou,~\IEEEmembership{Fellow,~IEEE}, and Nan~Yang,~\IEEEmembership{Senior Member,~IEEE}
\thanks{ Y. Li, X. Zhou and N. Yang are with the School of Engineering, the Australian National University, Canberra, ACT 2601, Australia (e-mail: yu.li1@anu.edu.au; xiangyun.zhou@anu.edu.au; nan.yang@anu.edu.au). Part of this work was conducted while M. Tatar Mamaghani (e-mail: milad.tatarmamaghani@gmail.com) was with the Australian National University. This work was supported by the Australian Research Council’s Discovery Projects (DP220101318).}
}

\maketitle

\begin{abstract}
Physical-layer characteristics, such as channel state information (CSI) and transmitter noise induced by hardware impairments, are often uniquely associated with a transmitter. This paper investigates transmitter anonymity at the physical layer from a signal design perspective. We consider an anonymous communication problem where the receiver should reliably decode the signal from the transmitter  but should not make use of the signal to infer the transmitter's identity. Transmitter anonymity is quantified using a Kullback-Leibler divergence (KLD)-based metric, which enables the formulation of explicit anonymity constraints in the precoder design. We then develop a transceiver design that preserves reliable communication under spatial multiplexing while preventing transmitter identification. The proposed framework employs a partitioned equal-gain combining (P-EGC) scheme that leverages receiver diversity without requiring transmitter-specific CSI. Simulation results demonstrate anonymity-reliability tradeoffs across different signal-to-noise ratios (SNRs), numbers of data streams, and alias selection methods. Moreover, the results reveal opposite trends of anonymity with respect to transmitter-dependent noise variations in the low-SNR and high-SNR regimes.
\end{abstract}

\section{Introduction}
\label{sec:introduction}

\IEEEPARstart{T}{he} proliferation of next-generation wireless networks is reshaping the landscape of communications by enabling a broad spectrum of Internet of Things (IoT) applications, including vehicular networks, uncrewed aerial vehicle (UAV) communications, and body sensor systems. The broadcast nature of wireless transmissions  inherently exposes such systems to various security threats such as eavesdropping, jamming, spoofing, and traffic analysis. These vulnerabilities are particularly critical in human-centric IoT services, where sensitive information, such as user identities, locations, and physiological data, is routinely exchanged, and unauthorized disclosure may lead to severe privacy violations and safety risks \cite{yang2017survey, sandeepa2022survey}.

Conventional upper-layer security mechanisms often face practical limitations in IoT environments, due to stringent constraints on power consumption, computational capability, and latency \cite{angueira2022survey}. In this context, physical-layer security (PLS) has emerged as an effective security enhancement by exploiting the intrinsic properties of wireless channels and transmit signals to provide lightweight and scalable protection. Under the general umbrella of PLS, two complementary objectives can be identified: \textit{confidentiality}, which safeguards the transmitted content from unauthorized access \cite{bloch_barros_2011, zhou2016physical, mamaghani2024performance}, and \emph{privacy}, which prevents adversaries---or even legitimate parties---from inferring sensitive user-related information, such as the identity or even the existence of the transmitter, from the received signals. Privacy-oriented techniques are generally developed under two different principles: anonymity aims to conceal the transmitter’s identity within a group of potential users \cite{wei2025physical}, whereas unobservability seeks to hide the very existence of transmissions, often through covert communications \cite{bash2015hiding}.

Transmitter anonymity is particularly critical in several emerging wireless applications. In vehicular systems, periodic beaconing enables safety-critical functions, such as navigation and collision avoidance,  yet roadside adversaries can exploit physical-layer signal characteristics to perform persistent vehicle identification and trajectory reconstruction, even in the presence of pseudonym-based protocols \cite{tzeng2015enhancing, lu2018survey}. Similar vulnerabilities arise in UAV communications, where transmitter identification can expose mission-critical information and enable tracking, jamming, or physical capture \cite{tedeschi2023privacy, fotouhi2019survey, pai2008confidentiality}. In body area networks, physical-layer identification by third-party base stations can further link sensitive physiological data to individual users, leading to severe privacy violations and potential commercial exploitation \cite{gope2015bsn}. In these scenarios, the physical-layer signals often carry device-specific information or location-dependent channel characteristics, which can be exploited for transmitter identification.  Therefore,  anonymity mechanisms operating above the physical layer may no longer be sufficient against adversaries capable of exploiting physical-layer features. This limitation motivates the development of physical-layer anonymity techniques that aim to mitigate identity leakage at the signal level \cite{weithepath}.

\subsection{Related Works}
We first note that the physical-layer anonymity problem is different from conventional physical-layer security problems focused on data confidentiality. Techniques such as artificial-noise injection and cooperative jamming are mainly designed to degrade eavesdroppers while preserving reliable reception at the legitimate receiver. In contrast, in physical-layer anonymous communication, the intended receiver is also the potential identity detector; therefore, confidentiality-oriented jamming mechanisms do not directly address the transmitter-identity leakage considered in this work.

Physical-layer anonymous communication is a new topic of research with only a handful of prior works. The existing physical-layer anonymity schemes can be broadly classified into \emph{strong-transmitter} and \emph{strong-receiver} architectures, depending on the relative number of transmit and receive antennas. In the strong-transmitter case, where the number of transmit antennas exceeds that of the receiver, reliable anonymous communication can be achieved solely through transmitter-side precoding  by exploiting the excess spatial degrees of freedom (DoF) \cite{wei2021fundamentals}. For this case, two classes of precoders have been developed: i) interference-suppression-based designs that employ transmitter-side phase equalization and ii) constructive-interference-based symbol-level precoders that steer signals into correct decision regions, thus eliminating the need for explicit equalization at either the transmitter or receiver. In contrast, the strong-receiver case, in which receive antennas outnumber transmit antennas, requires receiver-side combining to exploit the reception diversity as the transmitter-side precoding alone is insufficient to control the signal structure across all receive antennas \cite{wei2022physical1, wei2023phy}. An early design for this case employed channel-independent equal-gain combining (EGC) with a single data stream to preserve anonymity under heterogeneous transmitter antenna configurations \cite{wei2022physical1}. More recently, \cite{wei2023phy} investigated the multiplexing–diversity tradeoff and introduced alias-channel-based combiners to jointly exploit spatial multiplexing and reception diversity while maintaining anonymity. Specifically, the combiner was constructed from the averaged channel responses of the true transmitter and multiple alias users, implicitly treating them as equally likely during receiver-side transmitter detection. While these works differ in transceiver structure and signal processing design, they adopt channel state information (CSI) as the underlying feature for transmitter identification. Beyond CSI masking in \cite{weithepath,wei2021fundamentals, wei2022physical1, wei2023phy}, other physical-layer characteristics have been considered. For instance, joint masking of CSI and in-phase/quadrature (IQ) imbalance was explored to enhance physical-layer anonymity in \cite{cui2025physical}, which further highlights that CSI masking alone may be insufficient, since non-channel physical-layer characteristics can serve as additional identity-bearing features.

To quantify physical-layer anonymity, several performance metrics have been established in the literature. The most commonly used metric is the detection error rate (DER), which measures the probability that the receiver fails to identify the active user. A larger DER indicates a higher level of anonymity. Complementarily, anonymity entropy was proposed to characterize the receiver's uncertainty regarding the transmitter identity in a manner analogous to Shannon entropy \cite{cui2024closed, wei2022physical1}. More general analytical frameworks based on detection-error exponents and divergence-based measures were introduced in \cite{wei2025physical}, while closed-form DER expressions for specific systems were derived in \cite{cui2024closed, wei2023phy, wang2025closed}. These metrics provide a principled basis for analyzing anonymity.

 Despite the merits of the aforementioned works, two important limitations still remain in the existing physical-layer anonymous precoding literature. First, existing anonymous precoding designs assume \emph{homogeneous} effective noise variances across users. In contrast, heterogeneous noise variances caused by user-dependent hardware-induced distortion create additional statistical identity leakage even when the CSI-dependent signal mismatch is well masked. Second, existing transmitter identification models are \emph{symbol-unaware}, implicitly assuming that the adversary does not exploit decoded data symbols when performing identification. However, a receiver capable of reliable decoding can leverage symbol-dependent statistics to enhance identification performance. The impact of such symbol-aware transmitter identification on anonymous precoder design has not been investigated. Addressing these limitations is essential for developing physical-layer anonymity mechanisms that remain robust against information-rich adversarial receivers.

\subsection{Our Contributions}
To address the aforementioned limitations, this paper makes the following contributions:
\begin{itemize}

 \item We introduce a \emph{symbol-aware adversarial model} for physical-layer transmitter identification, in which a receiver exploits all available information, including reliably decoded data symbols, to enhance identification performance. Unlike existing physical-layer anonymity frameworks that rely on symbol-unaware detection models, we explicitly incorporate symbol-dependent statistics into the identification process and analyze their implications for the anonymous precoder design.

\item We propose a \emph{generalized physical-layer anonymity framework} that jointly accounts for CSI and transmitter-dependent noise variances. By explicitly modeling heterogeneous transmitter-dependent noise variances in both the identification model and the precoder design, the proposed framework addresses a critical and previously overlooked source of identity leakage and enables anonymity guarantees in more general network settings. We further develop a heterogeneous-noise-aware alias selection strategy to exploit the unequal receiver-side detectability induced by transmitter-dependent noises.

\item We develop a \emph{transceiver} design for a multi-user multiple-input-multiple-output (MU-MIMO) system that incorporates anonymity constraints based on the Kullback-Leibler divergence (KLD)-based separability metric induced by the likelihood comparison between transmitter hypotheses, while preserving reliable spatial multiplexing performance under symbol-aware identification. We consider a strong-receiver scenario and propose a \emph{partitioned equal-gain combining} (P-EGC) architecture that enables reliable multi-stream reception without requiring transmitter-specific CSI, alias-channel construction, or the assumption of equal likelihood among the true and alias transmitters. The proposed P-EGC depends only on the antenna configuration and the number of data streams, making it compatible with scenarios where users exhibit unequal detectability at the receiver and robust to heterogeneous transmitter-associated disturbance variances.

\end{itemize}
Overall, this work contributes a symbol-aware heterogeneous-noise-variance anonymity analysis and a corresponding user-agnostic transceiver design for physical-layer anonymous communication. The symbol-aware identification model addresses an important conceptual gap: although the receiver is expected to reliably decode the transmitted data, the possible use of decoded symbols for transmitter identification has not been incorporated into prior anonymity analysis. Moreover, by allowing transmitter-dependent noise variances, the proposed likelihood model is not governed solely by CSI-dependent residual mismatch, as in the homogeneous-noise, symbol-unaware frameworks reviewed above. This leads to a more general anonymity analysis in which user-dependent noise variances are explicitly accounted for in the receiver-side hypothesis comparison. Based on this model, the proposed design jointly supports transmitter-identity protection and user-agnostic reliable multi-stream reception.

\textit{Notation}:
 Boldface lowercase and uppercase letters denote vectors and matrices, respectively, while calligraphic uppercase letters denote sets. The operator $\operatorname{vec}(\cdot)$ denotes the vectorization of a matrix obtained by stacking its columns.  The operators $(\cdot)^H$,  $(\cdot)^\dagger$, $\|\cdot\|$,  and $\|\cdot\|_F$  denote the conjugate transpose, Moore--Penrose pseudo-inverse, Euclidean norm,  and Frobenius norm, respectively. Trace and expectation are denoted by $\operatorname{tr}(\cdot)$ and $\operatorname{E}[\cdot]$, respectively. The natural logarithm is denoted by $\ln(\cdot)$. In addition, $\mathcal{CN}(\mathbf{0},\mathbf{\Sigma})$ represents a circularly symmetric complex Gaussian  (CSCG) distribution with zero mean and covariance matrix $\mathbf{\Sigma}$, $\mathbf{I}_n$  represents the $n \times n$ identity matrix,  $\mathcal{O} (\cdot)$ represents the big-O notation, and $\Re\{\cdot\}$ and $\Im\{\cdot\}$ represent the real and imaginary parts, respectively.
\section{System Model}\label{sec:system}
We consider an uplink MU-MIMO system comprising $K$ transmitting users and a single receiver acting as a central server. Each user is equipped with $N_t$ transmit antennas, while the receiver employs $N_r$ receive antennas. We assume $N_r > N_t$, which is a typical uplink scenario and is referred to as the strong-receiver scenario in the physical-layer anonymous communication literature. This work focuses on uplink physical-layer anonymous transmission. Anonymous downlink transmission with user-specific responses, which may require higher-layer mechanisms, is beyond the scope of this study. The system operates under an anonymous collision-free time division multiple access (TDMA)-type access model, such that one user is active in each transmission block.

\subsection{Anonymous Access and Receiver Model}

In our considered system, the anonymous access and scheduling procedure is assumed to be handled at a higher layer. For example, the transmission blocks can be aligned with a common frame or slot structure, while the user access can be supported by randomized contention or transmitter-side coordination~\cite{song2012laa}. The resulting scheduling information is not revealed to the receiver.
Consequently, although exactly one user transmits in each transmission block, the receiver does not know the identity of this active user, which is consistent with the standard physical-layer anonymous communication setting \cite{weithepath, wei2021fundamentals, wei2022physical1, wei2023phy,cui2024closed,cui2025physical}. We note that the transmitter identity considered in this work refers to the physical-layer identity inferred from received signal characteristics, such as channel-dependent signal mismatch and transmitter-dependent noise variances, rather than protocol-level identifiers such as packet, session, or user IDs. Any identifiers required for packet association or data reassembly are assumed to be protected by higher-layer anonymous mechanisms, such as temporary pseudonyms, encrypted session identifiers, or otherwise unlinkable credentials \cite{camenisch2001efficient, davidson2018privacy}. Under this assumption, this work focuses on preventing additional identity leakage from physical-layer observations.

We adopt a semi-trustworthy receiver model, in which the receiver acts as a curious service provider that reliably decodes transmitted data while simultaneously attempting to infer the transmitter's identity. The purpose of physical-layer anonymous communication is to achieve the following two objectives simultaneously \cite{weithepath,wei2021fundamentals, wei2022physical1, wei2023phy,cui2024closed, cui2025physical, wei2025physical}:
\begin{itemize}
    \item To allow reliable decoding at the receiver based on the received signal.
    \item  To prevent the receiver from successfully inferring the transmitter’s identity based on the received signal.
\end{itemize}
Also, as in \cite{weithepath,wei2021fundamentals, wei2022physical1, wei2023phy,cui2024closed, cui2025physical}, it is assumed that all communication nodes have perfect knowledge of CSI. More specifically, the receiver knows the channel matrices of all users, but it does not know which candidate channel or equivalently which user is associated with the current transmission block.

Regarding the semi-trustworthy receiver, data decoding and transmitter identification are two different operations. Recall that the two objectives of physical-layer anonymous communication are to allow reliable decoding but prevent transmitter identification. Therefore, it is reasonable to consider the receiver first performing data decoding, which is required to be reliable, and then performing transmitter identification through multiple-hypothesis testing (MHT). The decoded data symbols do not carry an explicit active-user identifier, but they are used as side information by the receiver, i.e., symbol-aware, to evaluate the likelihood of the received signal under every transmitter-identity hypothesis during the transmitter identification process. This transmitter identification process will be described and analyzed in  Section~\ref{sec:detection}.

Therefore, our design for anonymous communication needs to take symbol-aware transmitter identification into account. From the communication design point of view, the active transmitter uses its CSI and data symbols to design the anonymous precoder, while the receiver applies the predetermined, user-independent combiner to decode the data without knowing the transmitter identity (or equivalently the actual channel matrix the received signal goes through). Such a design needs to meet the two objectives of physical-layer anonymous communication, as to be described in Section~\ref{sec:anony_design}.

\subsection{Signal Representation}

Let $k \in \mathcal{K}=\{1, 2, \cdots, K\}$ be the index of the scheduled user for transmission. Each user has its own data symbols and the scheduled user transmits its corresponding symbol matrix in the considered transmission block. We adopt a Rayleigh block-fading model in which the channel remains constant throughout a transmission block and varies independently across blocks.  We denote $\mathbf{H}_k \in \mathbb{C}^{N_r \times N_t}$ as the channel matrix from the \order{k} user to the receiver.
Let $N_s \leq N_t$ denote the number of data streams, and let $\mathbf{S}_k = [\mathbf{s}_k^1,\ldots,\mathbf{s}_k^L] \in \mathbb{C}^{N_s \times L} $ denote the transmitted symbol matrix of the \order{k} user over a  transmission block of length $L$. During the \order{\ell} symbol interval, the transmitted signal vector is generated as $\mathbf{x}_k^\ell = \mathbf{W}_k^\ell \mathbf{s}_k^\ell $, where $\mathbf{W}_k^\ell\in\mathbb{C}^{N_t\times N_s}$ denotes the symbol-level precoding matrix associated with $\mathbf{s}_k^\ell$. Accordingly, the transmitted signal block is given by $\mathbf{X}_k = [\mathbf{x}_k^1,\ldots,\mathbf{x}_k^L] \in\mathbb{C}^{N_t\times L} $.

The received signal corresponding to the \order{k} user is expressed as
\begin{equation}\label{Yk}
\mathbf{Y}_k = \mathbf{H}_k\mathbf{X}_k+\mathbf{N}_k,
\end{equation}
where $\mathbf{Y}_k\in\mathbb{C}^{N_r\times L}$ is the received signal  block and $\mathbf{N}_k \in \mathbb{C}^{N_r \times L}$ represents the equivalent aggregate disturbance observed at the receiver. To clarify the physical origin of this aggregate disturbance, consider the received signal during the \order{\ell} symbol interval. In the presence of residual transmitter hardware impairments, it can be expressed as $\mathbf{y}_k^\ell=\mathbf{H}_k \left(\mathbf{x}_k^\ell+\boldsymbol{\eta}_{t,k}^\ell \right) + \mathbf{n}_r^\ell = \mathbf{H}_k\mathbf{x}_k^\ell+\mathbf{z}_k^\ell $, where $\mathbf{z}_k^\ell =\mathbf{H}_k\boldsymbol{\eta}_{t,k}^\ell+\mathbf{n}_r^\ell$ denotes the physical aggregate disturbance at the receiver. Here, $\mathbf{n}_r^\ell\sim \mathcal{CN}(\mathbf{0},\sigma_r^2\mathbf{I}_{N_r})$ denotes the receiver  thermal noise, whose variance is common to all users, and $\boldsymbol{\eta}_{t,k}^\ell\sim \mathcal{CN}(\mathbf{0},\sigma_{t,k}^2\mathbf{I}_{N_t})$ denotes the residual hardware distortion associated with the \order{k} user \cite{studer2010mimo,bjornson2013hardware,zhang2014mimo}. The two disturbance terms are assumed to be mutually independent, independent across symbol intervals, and independent of the intended transmitted signal. Accordingly, conditioned on $\mathbf{H}_k$, the covariance of $\mathbf{z}_k^\ell$ is $\mathbf{R}_{z,k} =\sigma_{t,k}^2 \mathbf{H}_k\mathbf{H}_k^H +\sigma_r^2\mathbf{I}_{N_r} $. Hence, the transmitter-induced component is generally spatially colored for an instantaneous channel realization.

In this work, we focus on the \emph{aggregate disturbance-level heterogeneity} induced by different transmitter hardware-distortion levels and characterize this identity-bearing feature through the channel-averaged disturbance power. Under the i.i.d. Rayleigh fading model adopted in this work, $\mathbb{E}_{\mathbf{H}_k}[\mathbf{H}_k\mathbf{H}_k^H] =N_t\mathbf{I}_{N_r}$, and therefore $ \mathbb{E}_{\mathbf{H}_k} [\mathbf{R}_{z,k}] =\left(\sigma_r^2+N_t\sigma_{t,k}^2 \right)\mathbf{I}_{N_r} \triangleq \sigma_k^2\mathbf{I}_{N_r}$. Accordingly, $\sigma_k^2$ represents the equivalent aggregate disturbance variance associated with the \order{k} user, comprising the common receiver thermal-noise component $\sigma_r^2$ and a transmitter-specific contribution  determined by the residual hardware-distortion level $\sigma_{t,k}^2$. For analytical tractability, the subsequent analysis adopts an isotropic equivalent disturbance model matched to this channel-ensemble covariance, i.e., $\operatorname{vec}(\mathbf{N}_k)\sim \mathcal{CN}(\mathbf{0},\sigma_k^2\mathbf{I}_{N_rL})$.

\begin{remark}
Residual RF hardware impairments may persist after calibration and compensation, and are commonly modeled as additive distortion terms whose severity depends on the transceiver hardware quality \cite{studer2010mimo,bjornson2013hardware,zhang2014mimo}. Their practical impact on MIMO reception has also been demonstrated experimentally \cite{suzuki2008transmitter}. Motivated by these observations, this work focuses on the heterogeneity of aggregate disturbance levels across transmitters as a power-domain physical-layer identity feature. The instantaneous disturbance covariance additionally contains channel-dependent spatial structure, which provides a higher-dimensional identity feature. Jointly exploiting disturbance-power heterogeneity and such spatial covariance information constitutes an interesting extension of the considered identification framework.
\end{remark}

\section{Transmitter Detection Strategy}\label{sec:detection}
In this section, we investigate the receiver-side detection mechanism used to evaluate transmitter anonymity performance. We consider a scenario in which the receiver seeks to identify the active user based on the observed signal $\mathbf{Y}$ and the known CSI of all users. Without loss of generality, we assume that the actual transmitter is the \order{k} user. However, the transmitter identity is unknown to the receiver in the considered anonymous communication setting. For the receiver-side detection, we define the hypotheses of the received signal model using $\mathbf{Y}$ without a user-specific subscript.

Transmitter identification is performed once per transmission block based on the entire received signal matrix $\mathbf{Y}$. Using the decoded symbol matrix $\mathbf{S}_k$ as side information, the receiver performs blockwise fitting under each candidate channel hypothesis. Specifically, under hypothesis $\mathcal{H}_i$, the receiver uses a single matrix $\mathbf{W}_i\in\mathbb{C}^{N_t\times N_s}$ to fit the received signal matrix $\mathbf{Y}$ using the blockwise representation $\mathbf{H}_i\mathbf{W}_i\mathbf{S}_k$. Accordingly,  the transmitter-identification problem is formulated as the following MHT problem, with a blockwise fitting model employed under each candidate hypothesis:
\begin{align}\label{mht_prob}
\begin{cases}
\mathcal{H}_1 &:  \mathbf{Y} = \mathbf{H}_1 \mathbf{W}_1 \mathbf{S}_k + \mathbf{N}_1,\\
\mathcal{H}_2 &:  \mathbf{Y} = \mathbf{ H}_2 \mathbf{W}_2 \mathbf{S}_k + \mathbf{N}_2,\\
&\vdots  \\
\mathcal{H}_K &:  \mathbf{Y} = \mathbf{ H}_K \mathbf{W}_K \mathbf{S}_k + \mathbf{N}_K,
\end{cases}
\end{align}
where the hypothesis set $\pmb{\mathcal{H}}=\{\mathcal{H}_1, \mathcal{H}_2, \ldots, \mathcal{H}_K \}$ represents $K$ possible candidates for the originating transmitter and $\mathbf{H}_i$ denotes the channel matrix of the \order{i} candidate user, while $\mathbf{W}_i$ denotes the hypothesis-dependent blockwise-equivalent precoding matrix in receiver-side identification. It is distinguished from $\mathbf{W}_k^\ell$, which denotes the actual transmitter-side symbol-level precoder employed during the \order{\ell} symbol interval.

Two major methodologies exist for composite hypothesis testing: the Bayesian approach and the generalized likelihood ratio test (GLRT) \cite{kay1998detection}. For the MHT problem considered here, we adopt the GLRT framework, as it does not require prior distributions on the unknown blockwise-equivalent precoding matrices and is therefore more suitable and robust for anonymous communication scenarios. The first step in GLRT is to obtain the maximum-likelihood (ML) estimate of the unknown parameters in the linear model given in \eqref{mht_prob}. In the considered transmitter-identification problem, the receiver is assumed to know the CSI and the effective noise variance of each candidate user. We consider a symbol-aware identification model in which the decoded symbol matrix $\mathbf{S}_k$ of the active transmission block is used as known side information. The subscript $k$ indicates its origin in the analysis and does not represent identity information available to the receiver. In contrast, the blockwise-equivalent precoding matrix $\mathbf{W}_i$ under each candidate hypothesis is unknown to the receiver and is estimated within the GLRT. We highlight that data decoding does not require the receiver to know the actual symbol-level precoders, since the anonymous transceiver design in Section~\ref{sec:anony_design} enables user-agnostic demodulation without knowing the active user's identity or its precoders.

Assuming the actual transmitter is the \order{k} user, the hypothesis $\mathcal{H}_k$ is correct.
Under $\mathcal{H}_k$, the ML estimate of the blockwise-equivalent precoder is
$\mathbf{\widehat{W}}_k = \mathbf{H}_k^{\dagger}\mathbf{Y}\mathbf{S}_k^{\dagger}$,
where $\mathbf{H}_k^{\dagger}=(\mathbf{H}_k^H\mathbf{H}_k)^{-1}\mathbf{H}_k^H$ and, assuming that $\mathbf{S}_k$ has full row rank, i.e., $\operatorname{rank}(\mathbf{S}_k)=N_s$,
$\mathbf{S}_k^{\dagger}=\mathbf{S}_k^H(\mathbf{S}_k\mathbf{S}_k^H)^{-1}$.

Similarly, under an incorrect hypothesis $\mathcal{H}_i$, the estimated blockwise-equivalent precoder becomes $\mathbf{\widehat{W}}_i = \mathbf{H}_i ^{\dagger} \mathbf{Y} \mathbf{S}_k^{\dagger}$. Using these estimates, the reconstructed received signals under the correct and incorrect hypotheses  are given by
\begin{align}\label{Yk_hat}
\mathbf{\widehat{Y}}_k=\mathbf{ H}_k \mathbf{\widehat{W}}_k \mathbf{S}_k=\mathbf{H}_k\mathbf{X}_k
\mathbf{S}_k^\dagger\mathbf{S}_k + \mathbf{ H}_k \mathbf{ H}_k^\dagger \mathbf{N}_k \mathbf{S}_k^{\dagger} \mathbf{S}_k
\end{align}
and
\begin{align}\label{Yi_hat}
\mathbf{\widehat{Y}}_i=\mathbf{ H}_i \mathbf{ \widehat{W}}_i\mathbf{S}_k= \mathbf{H}_i \mathbf{H}_i^\dagger (\mathbf{H}_k\mathbf{X}_k + \mathbf{N}_k) \mathbf{S}_k^\dagger \mathbf{S}_k,
\end{align}
respectively.

Accordingly, the profile likelihoods obtained after blockwise fitting under the candidate hypotheses are given by
\begin{equation}\label{pdfs}
\begin{split}
\mathcal{H}_1 &:  p(\mathbf{Y}; \mathbf{\widehat{W}}_1,\mathcal{H}_1) = \frac{1}{(\pi \sigma_1^2)^{N_r L}} \exp\left( - \frac{1}{\sigma_1^2} \| \mathbf{Y} - \mathbf{\widehat{Y}}_1 \|_F^2 \right),\\
\mathcal{H}_2 &:  p(\mathbf{Y}; \mathbf{\widehat{W}}_2,\mathcal{H}_2) = \frac{1}{(\pi \sigma_2^2)^{N_r L}} \exp\left( - \frac{1}{\sigma_2^2} \| \mathbf{Y} - \mathbf{\widehat{Y}}_2 \|_F^2 \right),\\
&\vdots  \\
\mathcal{H}_K&\!: p(\mathbf{Y};\! \mathbf{\widehat{W}}_K,\mathcal{H}_K) \!= \!\frac{1}{(\pi \sigma_K^2)^{N_r L}} \exp\!\left( \!- \frac{1}{\sigma_K^2} \| \mathbf{Y} \!-\! \mathbf{\widehat{Y}}_K \|_F^2\! \right).
\end{split}
\end{equation}
Under the GLRT framework, the receiver selects hypothesis $\mathcal{H}_k$ if its likelihood $ p(\mathbf{Y}; \mathbf{\widehat{W}}_k,\mathcal{H}_k)$ is the largest among all candidates.
It is common to employ the log-likelihood as the test statistic due to its monotonicity and numerical stability. Accordingly, we define $ T_k = \ln p(\mathbf{Y}; \mathbf{\widehat{W}}_k,\mathcal{H}_k)$
as the test statistic under the hypothesis $\mathcal{H}_k$ such that the receiver decides in favor of the hypothesis associated with the maximal $T_k$.

\subsection{Anonymity Metric: Detection Error Rate}

In this paper, we evaluate transmitter anonymity through the probability that the receiver fails to correctly identify the active user. Consistent with existing studies on physical-layer anonymity (e.g., \cite{wei2023phy, cui2024closed}), we adopt the DER as the anonymity metric. To characterize the DER, we examine the statistical behavior of the detection test statistic $T$ under the correct and incorrect hypotheses.

Suppose the \order{k} user is the actual transmitter, implying that the true hypothesis is $\mathcal{H}_k$, while any $\mathcal{H}_i$ ($\forall i \in \mathcal{K}$ and $i \neq k$) represents an incorrect hypothesis. The corresponding  log-likelihood test statistics can be written as
\begin{subequations}\label{test_stats}
    \begin{align}
    T_k &= \ln p(\mathbf{Y}; \mathbf{\widehat{W}}_k,\mathcal{H}_k) = -N_r L \ln (\pi \sigma_k^2) - \frac{1}{ \sigma_k^2} \| \mathbf{\widehat{Y}}_k - \mathbf{Y} \|_F^2\\
\intertext{and}
    T_i &= \ln p(\mathbf{Y}; \mathbf{\widehat{W}}_i,\mathcal{H}_i) = -N_r L \ln (\pi \sigma_i^2) - \frac{1}{ \sigma_i^2} \| \mathbf{\widehat{Y}}_i - \mathbf{Y} \|_F^2.
    \end{align}
\end{subequations}

Since the receiver identifies the transmitting user based on the largest test statistic, an incorrect detection occurs when $T_i > T_k$ for some $i\neq k$. To characterize this event, we introduce the difference random variable $D_i = T_i - T_k$, which measures the \emph{relative likelihood} of an incorrect hypothesis compared with the true one. Therefore, the DER can be equivalently written as $\text{DER} = 1 - \Pr\big( D_i < 0,\ \forall i \neq k \big) $. Computing this probability requires the joint distribution of the set $\{D_i\}_{i \neq k}$. However, the variables $D_i$ are mutually dependent because each term includes the common term $T_k$. In addition, the test statistics $T_k$ and $T_i$ are quadratic forms of CSCG noise, and thus follow non-Gaussian generalized chi-square–type distributions. As a result, the joint cumulative distribution function (CDF) involves high-dimensional integrals over dependent, non-Gaussian random variables, rendering an exact closed-form expression for the DER analytically intractable. Nevertheless, the DER metric can be evaluated numerically through Monte-Carlo simulations.

\begin{remark}
If we approximate $T_k$ and $T_i$ as Gaussian distributed, then with both mean and variance specified, the distributions would be fully characterized, allowing a closed-form approximation of the anonymity metric. In this case, for each non-transmitting user $i$, the difference variable $D_i = T_i - T_k$ would also be Gaussian, and quantities such as $ \Pr(D_i > 0) = 1- \Pr(D_i <0)$ could be directly evaluated via the Gaussian CDF. If the  variables $\{D_i\}_{i \neq k}$ were further assumed independent, the overall anonymity metric can be expressed as $1- \prod_{i \neq k} \Pr(D_i < 0)$.

However, this approach relies on several strong assumptions and approximations. In practice, these assumptions introduce a considerable gap between the theoretical estimate and the actual behavior in our scenario. $\| \mathbf{\widehat{Y}}_k - \mathbf{Y} \|_F^2$ and $\| \mathbf{\widehat{Y}}_i - \mathbf{Y} \|_F^2$ in test statistics, and consequently $D_i$, are non-Gaussian because they involve quadratic forms of CSCG noise, and the variables $\{D_i\}_{i\neq k}$ are mutually dependent through the shared term $T_k$. These issues make it impossible to obtain the joint distribution of $\{D_i\}_{i \neq k}$ in closed form, making a closed-form evaluation of DER analytically intractable.
\end{remark}

\subsection{Anonymity Metric: Information-theoretic Metric}
To gain analytical insight into the factors influencing anonymity, we adopt an information-theoretic perspective.
Specifically, we examine a KLD-based separability metric constructed from the expected log-likelihood-ratio form of the KLD~\cite{cover1999elements}. Let $q_k(\mathbf{Y})$ denote the true received-signal probability density function (PDF) when the \order{k} user is the actual transmitter, as induced by the physical transmission model in \eqref{Yk}. Moreover, let $p_k(\mathbf{Y}) \triangleq p(\mathbf{Y}; \mathbf{\widehat{W}}_k,\mathcal{H}_k)$ and $p_i(\mathbf{Y}) \triangleq p(\mathbf{Y}; \mathbf{\widehat{W}}_i,\mathcal{H}_i)$ denote the likelihood values used by the GLRT under the true and false hypotheses, respectively. The metric is defined as
\begin{equation}
\label{KLdef}
\mathcal{D}_{k,i} = \operatorname{E}_{q_k} \left[
\ln\frac{p_k(\mathbf{Y})}{p_i(\mathbf{Y})}\right] =\int q_k(\mathbf{Y}) \ln\frac{p_k(\mathbf{Y})}{p_i(\mathbf{Y})} \,d\mathbf{Y},
\end{equation}
where the expectation is taken with respect to (w.r.t.) the true received-signal distribution $q_k(\mathbf{Y})$ under hypothesis $\mathcal{H}_k$. Since the likelihoods are evaluated using the ML estimates of the unknown blockwise-equivalent precoding matrices under different hypotheses, this quantity is interpreted as a profile-likelihood-based separability metric induced by the GLRT statistic. It quantifies how distinguishable the received signals are between the two cases, i.e., a larger value indicates that the two hypotheses are more distinguishable, making correct detection easier for the receiver. In the following, we present Lemma \ref{lemma_kld} to analytically quantify this KLD-based metric in our anonymous communication system.

\begin{lemma}\label{lemma_kld}
The KLD-based metric for hypotheses $\mathcal{H}_k$ and $\mathcal{H}_i$, defined in \eqref{KLdef}, admits the closed-form expression given by
\begin{align}\label{kld}
\mathcal{D}_{k,i}
&= \operatorname{E}[T_k]-\operatorname{E}[T_i]
= N_rL\ln\left(\frac{\sigma_i^2}{\sigma_k^2}\right) \nonumber\\
&\quad+\left(1-\frac{\sigma_k^2}{\sigma_i^2}\right)(N_sN_t-N_rL) \nonumber\\
&\quad+\left(\frac{1}{\sigma_i^2}-\frac{1}{\sigma_k^2}\right)
\left\|\mathbf{H}_k\mathbf{X}_k
\left(\mathbf{S}_k^\dagger\mathbf{S}_k-\mathbf{I}_L\right)\right\|_F^2 \nonumber\\
&\quad+\frac{1}{\sigma_i^2}
\left\|\left(\mathbf{H}_i\mathbf{H}_i^\dagger-\mathbf{I}_{N_r}\right)
\mathbf{H}_k\mathbf{X}_k\mathbf{S}_k^\dagger\mathbf{S}_k\right\|_F^2.
\end{align}
\end{lemma}
\begin{proof}
    See Appendix \ref{Appendix_A}.
\end{proof}

Recalling that $D_i = T_i - T_k$, the KLD-based metric in Lemma~\ref{lemma_kld} can be written as $\mathcal{D}_{k,i}=-\,\operatorname{E}[D_i]$. Thus, a larger expected value $\operatorname{E}[D_i]$ corresponds to a smaller value of this metric, meaning that the two hypotheses become less distinguishable. Although the DER is evaluated empirically in Section~\ref{sec:simulations} due to the lack of a closed-form expression, this metric serves as an analytical indicator of hypothesis separability, where a lower metric value generally indicates lower separability and hence a higher DER. If we define each user other than the true transmitter as an alternative user, then the pairwise anonymity between the true transmitter and any alternative user $i$ can be characterized by $\operatorname{E}[D_i]$.

\subsection{Transmitter Anonymity Constraint}
\label{sec:trans_anony}
In this subsection, we quantify an anonymity constraint for the transmitter detection strategy.

Here,  we refer to the non-transmitting users (indexed by $\forall i\ne k$) as the \emph{alias} users whose channels are intentionally exploited by the transmitter to disguise the true transmission. We first define the alias-candidate set for the \order{k} user as $\mathcal{A}^{\mathrm{cand}}_k = \{1,\ldots,K\}\setminus\{k\}$, which contains all non-transmitting users whose channels may be exploited for anonymity. From the candidate set $\mathcal{A}^{\mathrm{cand}}_k$, the transmitter selects a subset $\mathcal{A} \subseteq \mathcal{A}^{\mathrm{cand}}_k$ as the index set of the selected alias users, where $|\mathcal{A}| = a, 1\le a\le K-1$, represents the number of aliases selected by the transmitter for precoder design. We note that the receiver selects the user corresponding to the largest test statistic $T_i$. To achieve pairwise anonymity between the true user and each alias, the likelihood of the received signal $\mathbf{Y}$ under the alias hypothesis $\mathcal{H}_i$ should be comparable to that under the true hypothesis $\mathcal{H}_k$. From Lemma \ref{lemma_kld}, we see that the pairwise anonymity between the true user and an alias user can be characterized by $\operatorname{E}[D_i]$, where a larger $\operatorname{E}[D_i]$ indicates smaller separability between the two hypotheses. To enhance anonymity, one can minimize this separability metric, which is equivalent to maximizing its negative value, $\operatorname{E}[D_i]=\operatorname{E}[T_i] -\operatorname{E}[T_k]$. To achieve this aim, $\operatorname{E}[D_i]$ can be expressed as
\begin{equation}\label{EDi}
\operatorname{E}[D_i] = E_{\text{const}} +\left(\frac{1}{\sigma_k^2} - \frac{1}{\sigma_i^2} \right)\eta_k-\frac{1}{\sigma_i^2}\Delta,
\end{equation}
where
\begin{align}
  E_{\text{const}} =  N_r L \ln\left(\frac{\sigma_k^2}{\sigma_i^2}\right) +\left( 1 - \frac{\sigma_k^2}{\sigma_i^2}\right) \left( N_r L - N_s N_t \right),
\end{align}
\begin{equation}\label{eq:eta}
\eta_k = \left\| \mathbf{H}_k\mathbf{X}_k \left(\mathbf{S}_k^\dagger\mathbf{S}_k -\mathbf{I}_L\right) \right\|_F^2,
\end{equation}
and
\begin{equation}\label{eq:delta}
\Delta =\left\|\left(\mathbf{H}_i\mathbf{H}_i^\dagger-\mathbf{I}_{N_r}\right)\mathbf{H}_k\mathbf{X}_k\mathbf{S}_k^\dagger\mathbf{S}_k\right\|_F^2.
\end{equation}

We note that \eqref{EDi} separates the expected pairwise log-likelihood difference into three parts. The first part, $E_{\text{const}} $, depends only on the noise variances of the true user and the alias user, and is independent of the precoder. In the second part, $\eta_k$ represents the blockwise fitting residual energy common to all hypotheses. The third part, $\Delta/\sigma_i^2$, is the alternative-channel-dependent mismatch term that can be shaped by the precoder. Therefore, for a fixed $\eta_k$ and alias user $i$, maximizing the pairwise anonymity measure $\operatorname{E}[D_i]$ is equivalent to minimizing $\Delta$, which can be upper-bounded as
\begin{equation}
\label{constraint}
\begin{split}
\Delta &=\left\|\left(\mathbf{I}_{N_r}-\mathbf{H}_i\mathbf{H}_i^\dagger\right)\mathbf{H}_k\mathbf{X}_k\mathbf{S}_k^\dagger\mathbf{S}_k\right\|_F^2\\
&\stackrel{\mathrm{(a)}}{\le}\left\|\left(\mathbf{I}_{N_r}-\mathbf{H}_i\mathbf{H}_i^\dagger\right)\mathbf{H}_k\mathbf{X}_k\right\|_F^2\\
&= \sum_{\ell=1}^{L}\left\|\left(\mathbf{I}_{N_r}-\mathbf{H}_i\mathbf{H}_i^\dagger\right)\mathbf{H}_k\mathbf{W}_k^\ell\mathbf{s}_k^\ell \right\|_2^2,
\end{split}
\end{equation}
where $(a)$ follows from the fact that $\mathbf{S}_k^\dagger\mathbf{S}_k$ is the orthogonal projector onto the row space of $\mathbf{S}_k$.

Since $\Delta \ge 0$, driving this term toward zero provides the most favorable condition for improving anonymity, as it minimizes the distinguishability between the true and alias hypotheses. However, enforcing $\Delta = 0$ is overly restrictive and would impose a severe constraint on the precoder, compromising communication reliability. To balance anonymity with communication quality, we introduce a small relaxation parameter $\epsilon$ to allow a controlled, non-zero mismatch, with smaller values enforcing stronger anonymity. The anonymity constraint at the symbol level is thus given by
\begin{equation}\label{eq:cst_anonym}
\big\|(\mathbf{I}_{N_r} - \mathbf{H}_i \mathbf{H}_i^{\dagger}) \mathbf{H}_k \, \mathbf{W}_k^{\ell} \mathbf{s}_k^{\ell} \big\|_2^2 \;\le\; \epsilon,
\quad \forall i\ne k, \,\ell
\end{equation}
where $\mathbf{s}_k^{\ell}$ denotes the transmit symbol vector of the \order{k} user at timeslot~$\ell$. According to \eqref{constraint}, satisfying \eqref{eq:cst_anonym} over the $L$ symbol intervals guarantees $\Delta\leq L\epsilon$.

\begin{remark}
The condition in~\eqref{eq:cst_anonym} follows the theoretical definition of pairwise anonymity, which is formulated w.r.t. any non-transmitting user ($\forall i\ne k$). This definition specifies how pairwise anonymity is ensured between the true user and an arbitrary non-transmitting user, but does not require that the constraint be enforced for all such users in practice. Since imposing all $K\!-\!1$ pairwise constraints would overly restrict the precoding space and noticeably degrade communication reliability, the subsequent design enforces this constraint only on a selected alias subset $\mathcal{A}$.
\end{remark}

\section{Anonymous Communication Design}\label{sec:anony_design}

In a physical-layer anonymous communication framework, the transceiver design needs to support two objectives: reliable data decoding at the receiver and protection of the active transmitter's physical-layer identity.
To enable user-agnostic signal combining and reliable demodulation, we employ a constructive-interference (CI) precoder at the transmitter. The CI precoder pre-aligns the transmitted symbol phases so that, after spatial combining, the resulting received signal can be directly demodulated without requiring any CSI-dependent equalization at the receiver. In the following, we first review the CI precoding principles (see \cite{masouros2015exploiting} and references therein for further details) and then present our receiver combining scheme and anonymous precoder design tailored for the considered system.

\subsection{Constructive Interference Precoder}\label{sec:precoder}

In conventional systems, phase equalization is essential for exploiting spatial diversity in the sense that without aligning the signal phases across antennas, the received components cannot be coherently combined. However, in the anonymous setting, the receiver cannot construct the corresponding CSI-based equalizer without knowing the active channel. Specifically, in the strong-receiver scenario  ($N_r > N_t$), where the transmitter has limited spatial DoF, the transmitter cannot independently control the phase of each receive antenna. By pre-adjusting the transmitted signal phases, the CI precoder fulfills this role at the transmitter side, ensuring that the post-combined received symbols for each data stream remain correctly aligned for demodulation while preserving user anonymity. In other words, the CI precoding aims to actively shape the multi-stream interference inasmuch as it reinforces, rather than distorts, the desired signal components. Instead of suppressing interference, the precoder exploits the known symbol phases and the available CSI to steer the superimposed signal at the receiver toward the constructive region of the constellation. As a result, the post-combined received symbols are naturally aligned with their target directions, enhancing detection reliability without requiring CSI-based phase equalization. This approach effectively converts interference power into useful signal energy, improving both energy efficiency and noise resilience. Therefore, the CI precoding plays a crucial role in enabling anonymous communication under the strong-receiver scenario ($N_r >N_t$). We note that CI precoding is symbol-level in the sense that the transmitter designs the precoder according to the instantaneous transmitted symbol vector, whereas the symbol-aware identification model in Section~\ref{sec:detection} refers to the receiver-side use of the decoded symbol matrix for transmitter identification. We will design the appropriate CI precoder w.r.t. a new receiver combiner structure, introduced in the sequel, for anonymous communication.

\subsection{Partitioned Equal-Gain Combiner}\label{sec:pegc}

Conventional combining schemes such as a maximum-ratio combiner, widely used in non-anonymous MIMO settings, are not suitable for anonymous communication, because they inherently rely on user-specific CSI. In light of this, we propose a \emph{partitioned equal-gain combining} (P-EGC) scheme, which is applied at the receiver to exploit spatial diversity without relying on user-specific CSI. Specifically, the key idea of our P-EGC scheme is to partition receive antennas into several groups according to the number of transmitted data streams, then combine the signals received within each group corresponding to the same data stream with equal weights, while avoiding any combining across different antenna groups. This approach enables the receiver to benefit from spatial diversity in a user-agnostic manner, without relying on the transmitter's identity or channel characteristics.

The detailed combining process can be described as follows.  For any arbitrary number of receive antennas $N_r$, the antennas are divided into $N_s$ disjoint subsets such that $N_r > N_s$\footnote{
Since $N_s$ data streams are sent using $N_t$ transmit antennas, in the simplest case ($N_s = 1$), where the transmitter sends only one data stream, the receiver naturally achieves full diversity without the need to partition the receive antennas into subsets. As such, the precoder essentially functions as a vector ($\mathbf{w}_k \in \mathbb{C}^{N_t \times 1}$). However, considering multiplexing gain at the transmitter, the receiver must partition the receive antennas into subsets based on the number of multiplexed data streams. If the number of receive antennas is not an integer multiple of the number of streams, the antenna allocation among data streams may vary slightly, but this does not affect our design process.}. Let $\mathcal{S}_m$ denote the subset of receive antennas assigned to the \order{m} data stream with $|\mathcal{S}_m|$ being its cardinality. We denote $y_{k,j}^{\ell}$ as the received signal at the \order{j} antenna, where $j \in \mathcal{S}_m$, from the \order{k} user at timeslot $\ell$, which can be expressed as
\begin{equation}\label{ykj}
y_{k, j}^{\ell} = \mathbf{h}_{k, j} \sum_{n=1}^{N_s} \mathbf{w}_{k, n}^{\ell} s_{k, n}^{\ell} + n_{k, j}^{\ell},
\end{equation}
where $\mathbf{h}_{k,j} \in \mathbb{C}^{1 \times N_t}$ is the channel vector from the \order{k} transmitter to the \order{j} receive antenna in $\mathcal{S}_m$, and $\mathbf{w}_{k,n}^{\ell} \in \mathbb{C}^{N_t \times 1}$ is the precoder for the \order{n} data stream at timeslot $\ell$.
Let $s_{k,n}^{\ell} \in \mathbb{C}$ denote the symbol of the \order{n} data stream transmitted at timeslot $\ell$ and $n_{k,j}^{\ell} \sim \mathcal{CN}(0, \sigma_k^2)$ denote the noise in this received signal at the \order{j} receive antenna. Applying the P-EGC to all the received signals from $\mathcal{S}_m$, the combined signal can be expressed as
\begin{align}\label{r_km}
    r_{k, m}^{\ell} &= \sum_{j \in \mathcal{S}_m} y_{k,j}^{\ell}\nonumber\\
    &= \mathbf{c}_m \mathbf{H}_{k,m} \mathbf{W}_{k}^{\ell} \mathbf{s}_k^{\ell} + \tilde{n}_{k,m}^{\ell},\quad\forall k, m, \ell,
\end{align}
where $\mathbf{H}_{k,m} =
\begin{bmatrix}
\,\mathbf{h}_{k,j} \,
\end{bmatrix}_{j \in \mathcal{S}_m}
\in \mathbb{C}^{|\mathcal{S}_m| \times N_t}$  is the submatrix of the full channel $\mathbf{H}_k$ corresponding to antennas in $\mathcal{S}_m$, and $\mathbf{c}_m \in \mathbb{C}^{ 1 \times |\mathcal{S}_m|}$ denotes the equal gain combiner for the \order{m} data stream, which is assumed to be an all-ones row vector, i.e., $\mathbf{c}_m = \mathbf{1}^{ 1 \times |\mathcal{S}_m|}$.  In addition, $\mathbf{W}_k^\ell = [\mathbf{w}_{k,1}^\ell,\mathbf{w}_{k,2}^\ell,\ldots,\mathbf{w}_{k,N_s}^\ell]\in \mathbb{C}^{N_t \times N_s}$ is the precoder for all data streams at timeslot $\ell$, $\mathbf{s}_k^{\ell}$ denotes the transmit symbol vector across all data streams,  and $\tilde{n}_{k,m}^{\ell}\triangleq \mathbf{c}_m\mathbf{n}_{k, m}^{\ell}$ indicates the aggregated noise after combining, where $\mathbf{n}_{k,m}^{\ell} = [\, n_{k,j}^{\ell} \,]_{j \in \mathcal{S}_m}
\in \mathbb{C}^{|\mathcal{S}_m| \times 1}$ is the noise vector corresponding to the antenna subset~$\mathcal{S}_m$. Each $n_{k,j}^{\ell}$ is modeled as an independent and identically distributed (i.i.d.) CSCG random variable. Thus, the effective noise for each data stream follows $\tilde{n}_{k,m}^{\ell}\sim\mathcal{CN}(0,\sigma_{\mathrm{eff}}^2)$, such that $ \sigma_{\mathrm{eff}}^2=\operatorname{E}(|\tilde{n}_{k,m}^{\ell}|^2)
    = \|\mathbf{c}_m\|_2^2 \sigma_k^2
    = |\mathcal{S}_m|\, \sigma_k^2$.

For all $N_s$ data streams, the P-EGC can be expressed in a compact block-diagonal form as
\begin{equation}
\label{eq:comb_matrix}
    \mathbf{C}=
    \begin{bmatrix}
    \mathbf{c}_{1} & 0 & \cdots & 0 \\
    0 & \mathbf{c}_{2} & \cdots & 0 \\
    \vdots & \vdots & \ddots & \vdots \\
    0 & 0 & \cdots & \mathbf{c}_{N_s}
    \end{bmatrix},
\end{equation}
where $\mathbf{C}\in\mathbb{C}^{N_s\times N_r}$ depends only on the antenna grouping and is independent of the timeslot index $\ell$. The post-combined received signal vector at timeslot~$\ell$ can be expressed as
\begin{equation}
    \label{combined signal}
    \begin{split}
        \Tilde{\mathbf{r}}_k^{\ell} &= \mathbf{C}\left(\mathbf{H}_k \mathbf{W}_k^{\ell} \mathbf{s}_k^{\ell} + \mathbf{n}_k^{\ell}\right),\quad\forall k, \ell,
    \end{split}
\end{equation}
where $\Tilde{\mathbf{r}}_k^{\ell} = [\, r_{k,1}^{\ell}, \ldots, r_{k,N_s}^{\ell} \,]^T \in \mathbb{C}^{ N_s \times 1}$ denotes the vector of post-combined outputs for all streams and $\mathbf{n}_k^{\ell} = [\, n_{k,1}^{\ell}, \ldots, n_{k,N_r}^{\ell} \,]^T
\in \mathbb{C}^{N_r \times 1}$ is the received noise vector.

\subsection{Anonymous Precoding Requirements}\label{sec:precoding}
Here we specify the CI requirement used in the anonymous precoder design. By combining the CI principle with the P-EGC model introduced earlier, we express the geometric condition that the post-combined received symbol must satisfy in order to achieve constructive interference, which will be used as a constraint in the subsequent design problem.

We consider the quadrature phase shift keying (QPSK) modulation for the signal vector as a representative example of constellations commonly adopted in symbol-level CI precoding studies~\cite{masouros2015exploiting,wei2021fundamentals,wei2019multi,wei2020secure}. In QPSK, each transmitted symbol lies in one of the four quadrants of the complex plane, and the decision boundaries are the real and imaginary axes.
The goal of CI precoding is to ensure that the post-combined received signal of each data stream is pushed deeper into the correct quadrant, away from these decision boundaries\footnote{Although our work specifically considers the QPSK modulation, the CI precoding framework can be generalized to higher-order constellations such as M-QAM and M-PAM, by redefining the constructive regions around constellation points~\cite{li2020interference, wei2024sub}. Such extensions are left for future work.}.

Following our P-EGC receiver model, outlined in Section \ref{sec:pegc}, the combined received signal \eqref{r_km} can be rewritten as
\begin{equation}
r_{k, m}^{\ell} = \pmb{\alpha}_{k,m}^{\ell}\mathbf{s}_k^{\ell} + \tilde{n}_{k,m}^{\ell}, \quad\forall k, m, \ell,
\end{equation}
where $\pmb{\alpha}_{k,m}^{\ell}\triangleq\mathbf{c}_m \mathbf{H}_{k,m} \mathbf{W}_{k}^{\ell}$.
To ensure correct detection, the signal component $\pmb{\alpha}_{k,m}^{\ell}\mathbf{s}_k^{\ell}$ (ignoring the noise) must lie within the decision region.
Geometrically, for a QPSK symbol $s_{k,m}^{\ell}$ located on the unit circle with phase $\phi_{k,m}^{\ell}$, the constructive region corresponds to the angular sector $\phi_{k,m}^{\ell}\pm\theta$, where $\theta=\pi/4$. In the CI scheme, the constellation is properly rotated so that the desired symbol lies on the positive real axis. This can be equivalently achieved by multiplying the noise-less term $\pmb{\alpha}_{k,m}^{\ell}\mathbf{s}_k^{\ell}$ with the complex-conjugate of the symbol, subject to the following conditions:
\begin{subequations}\label{ci_cst}
\begin{align}
&\Re\left\{\pmb{\alpha}_{k,m}^{\ell}\mathbf{s}_k^{\ell}(s_{k,m}^{\ell})^*\right\}\ge\gamma,\quad\forall k, m,\label{cst1}\\
&\big|\Im\left\{\pmb{\alpha}_{k,m}^{\ell}\mathbf{s}_k^{\ell}(s_{k,m}^{\ell})^*\right\}\big| \nonumber\\
& \hspace{7mm}\!\le\! \tan\theta \big(\Re\left\{\pmb{\alpha}_{k,m}^{\ell}\mathbf{s}_k^{\ell}(s_{k,m}^{\ell})^*\right\}\!-\!\gamma \big),\quad\forall k, m,\label{cst2}
\end{align}
\end{subequations}
where $\gamma>0$ denotes the required minimum in-phase distance from the origin (a.k.a. constructive margin, which is proportional to the Euclidean decision margin), representing the guaranteed constructive amplitude of the received symbol. The real part in \eqref{cst1} corresponds to the in-phase projection of the rotated received symbol, i.e., the component that pushes the symbol away from the decision boundaries and thus contributes to reliable detection.
Meanwhile, \eqref{cst2} confines the symbol within the constructive angular sector. By noting that the right-hand side of \eqref{cst2} is nonnegative only when $\Re\{\cdot\} \ge \gamma$,  we find that \eqref{cst1} is satisfied when \eqref{cst2} holds.
Therefore, the two inequalities can be combined into a single constraint that jointly captures both the IQ requirements as
\begin{equation}\label{ci_cst_single}
\begin{aligned}
&\big|\Im\left\{
\pmb{\alpha}_{k,m}^{\ell}\mathbf{s}_k^{\ell}
(s_{k,m}^{\ell})^*
\right\}\big| \\
&\qquad \le
\tan\theta
\left(
\Re\left\{
\pmb{\alpha}_{k,m}^{\ell}\mathbf{s}_k^{\ell}
(s_{k,m}^{\ell})^*
\right\}
-\gamma
\right),
\quad \forall k,m.
\end{aligned}
\end{equation}
Constraint \eqref{ci_cst_single} ensures that the noiseless post-combined received symbol $\pmb{\alpha}_{k,m}^{\ell}\mathbf{s}_k^{\ell}$ lies in the constructive region of its target constellation point, making multi-stream interference beneficial for detection.

\begin{figure}[!t]
	\centering
\includegraphics[width=0.8\linewidth]{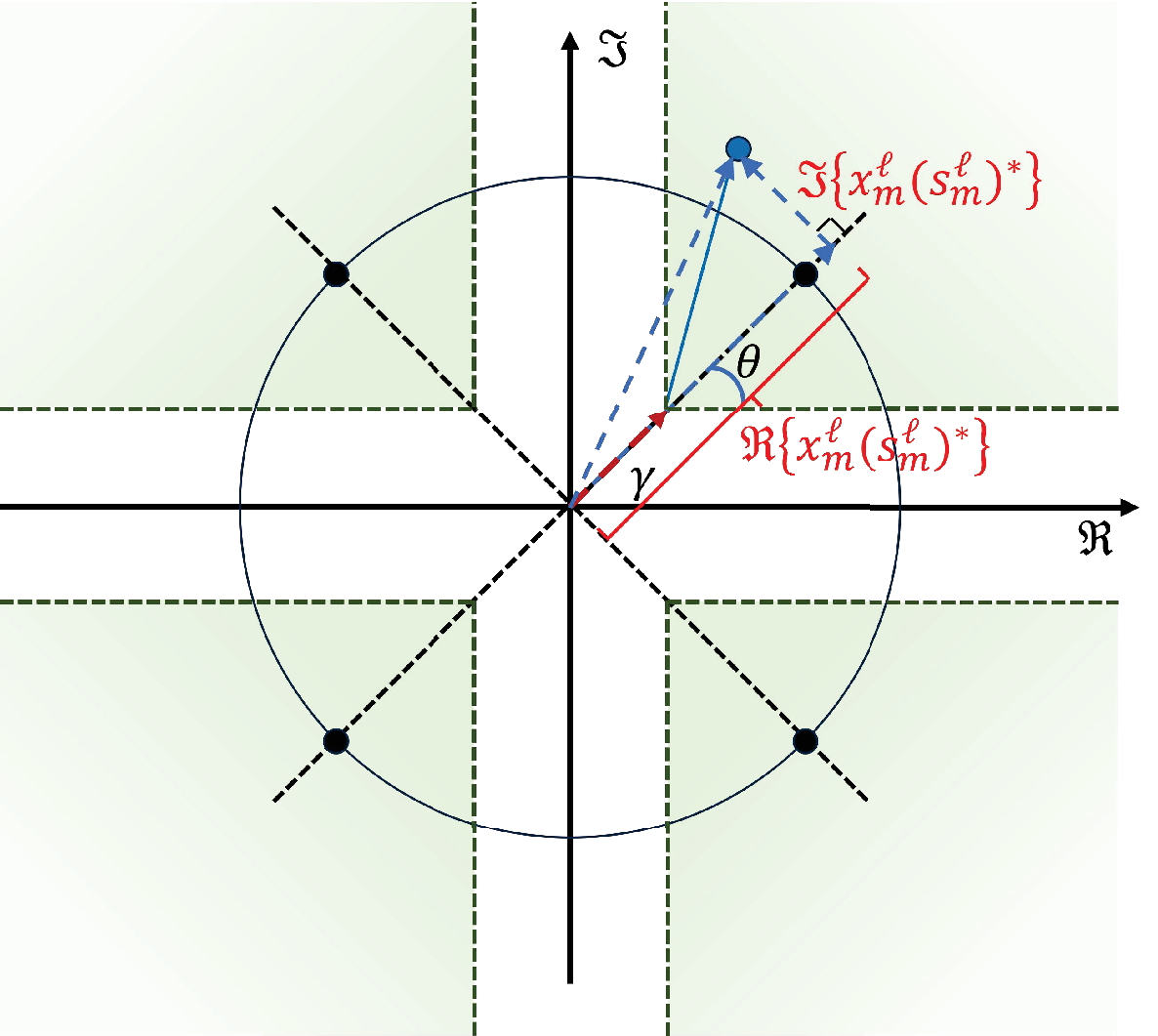}
	\caption{CI precoding for QPSK. The green area represents the constructive decision region and the black points denote the reference constellation symbols. The parameter $\gamma$ indicates the minimum in-phase distance from the origin. The blue point $x_m^{\ell}$ is the received symbol on stream~$m$ at timeslot~$\ell$, with target symbol $s_m^{\ell}$.  After rotating the point by the negative phase of the first-quadrant reference symbol, the rotated point lies in the constructive region if it satisfies the CI constraint in \eqref{ci_cst_single}, where $\theta = 45^\circ$ for QPSK.
	}
	\label{fig:CI}
\end{figure}

An illustrative example of the above CI precoding applied to QPSK is provided in Fig.~\ref{fig:CI}. In this figure, the blue point represents an example of the received symbol,  denoted by $x_m^{\ell}$. The dashed \(45^\circ\) line indicates the direction of the reference constellation point in the first quadrant, along which CI aims to push the received symbol into the constructive region. After rotating the received signal by the negative phase of the reference symbol, this \(45^\circ\) line aligns with the real axis. In this rotated domain, the real part corresponds to the component along the desired symbol direction, whereas the imaginary part represents the deviation from that direction. To ensure that the received symbol lies within the constructive region, the geometric relationship between its real and imaginary parts must satisfy \eqref{ci_cst_single}.

With the CI constraint established, we now turn to how communication reliability is quantified under this framework.
Under the CI precoding framework, multi-stream interference is aligned with the desired symbol direction and thus becomes a useful power contribution. Accordingly, the post-combining signal-to-disturbance ratio (SDR) for the \order{m} data stream of  the \order{k} user can be calculated as
\begin{equation}\label{eq:snr}
    \begin{split}
\Gamma_{k,m} &=\frac{\|\sum_{j \in \mathcal{S}_m} \mathbf{h}_{k,j} (\sum_{n=1}^{N_s} \mathbf{w}_{k,n}^{\ell} \mathbf{s}_{k,n}^{\ell})\|^2}{\mathbb{E}\!\left[\left|\sum_{j\in\mathcal{S}_m} n_{k,j}^{\ell}\right|^2\right] }\\
&=\frac{\|\mathbf{c}_m\mathbf{H}_{k,m}\mathbf{W}_k^{\ell} \mathbf{s}_k^{\ell} \|^2}
    {|\mathcal{S}_m|\sigma_k^2},\quad \forall k, m.
    \end{split}
\end{equation}
This expression characterizes the resulting post-combining reliability, whereas the optimization is carried out w.r.t. the guaranteed threshold \(\Gamma\) through the CI margin \(\gamma\).

\subsection{Design Problem and Solution}
Building on the proposed P-EGC scheme at the receiver and the CI precoding at the transmitter, we now focus on designing anonymous precoders $\{\mathbf{W}_k^{\ell}\}_{\ell=1}^{L}$ to jointly enhance communication reliability and enforce transmitter anonymity. For each timeslot $\ell$, we introduce $\Gamma^{\ell}$ as a \emph{common worst-stream} SDR guarantee, i.e., the SDR of
every data stream is constrained to be no smaller than $\Gamma^{\ell}$. To ensure uniform performance over time, we maximize the worst-timeslot guarantee, which yields the following max--min formulation:
\begin{subequations}\label{eq:P0_maxmin}
\begin{align}
\mathrm{P0}:~ &\max_{\{\mathbf{W}_k^{\ell},\,\Gamma^{\ell}\}_{\ell=1}^{L},\,\Gamma}\ \Gamma\notag\\
\text{s.t. }~
& \frac{\left\|\mathbf{c}_m\mathbf{H}_{k,m}\mathbf{W}_k^{\ell}\mathbf{s}_k^{\ell}\right\|^2}
{|\mathcal{S}_m|\sigma_k^2}\ \ge\ \Gamma^{\ell},\quad \forall m,\ \forall \ell, \label{eq:P0_sdr}\\
& \left| \Im \left\{ \mathbf{c}_m\mathbf{H}_{k,m}\mathbf{W}_k^{\ell}\mathbf{s}_k^{\ell}(s_{k,m}^{\ell})^* \right\} \right|
\leq \tan \theta \nonumber\\
&\hspace{-3mm}\times\left( \Re \left\{ \mathbf{c}_m\mathbf{H}_{k,m}\mathbf{W}_k^{\ell}\mathbf{s}_k^{\ell}(s_{k,m}^{\ell})^* \right\}
- \sqrt{\Gamma^{\ell}\sigma_{\mathrm{eff}}^2} \right),~\forall m, \ell, \label{eq:P0_ci}\\
&  \left\| \mathbf{W}_k^{\ell} \mathbf{s}_k^{\ell} \right\|^2_2 \leq  \frac{p_{\mathrm{tot}}}{L}, \quad \forall\ell, \label{eq:P0_pow}\\
&  \| (\mathbf{I}_{N_r} - \mathbf{H}_i \mathbf{H}_i^{\dagger})\mathbf{H}_k\mathbf{W}_k^{\ell}  \mathbf{s}_k^{\ell} \|_2^2
\leq \epsilon,\quad \forall i\in \mathcal{A},\ \forall \ell, \label{eq:P0_anon}\\
& \Gamma^{\ell} \ge \Gamma,\quad \forall \ell, \label{eq:P0_link}
\end{align}
\end{subequations}
where \eqref{eq:P0_sdr} ensures the SDR constraints for communication reliability, \eqref{eq:P0_ci} enforces CI conditions at the symbol level, ensuring that the received signals lie within the constructive region of the modulation constellation,  \eqref{eq:P0_pow} imposes the power budget for each symbol interval, where $p_{\mathrm{tot}}$ is the total transmit power, and \eqref{eq:P0_anon} incorporates the transmission anonymity constraint, determined in \eqref{eq:cst_anonym}.  Problem~$\mathrm{P0}$ implements a two-level fairness criterion, guaranteeing uniformly reliable transmission across both data streams $\forall m$ and timeslots $\forall \ell$.

Since all constraints in \eqref{eq:P0_sdr}--\eqref{eq:P0_link} are imposed independently per timeslot, and the
precoder $\mathbf{W}_k^{\ell}$ appears only in constraints indexed by the same $\ell$, the problem $\mathrm{P0}$ admits a per-timeslot decomposition.  Accordingly, for a given timeslot $\ell$, we can solve the following subproblem (dropping the index $\ell$ for notational simplicity):
\begin{subequations}\label{eq:P1}
\begin{align}
\mathrm{P1}:~ &\max_{\mathbf{W}_k,\,\Gamma}\ \Gamma \notag\\
\text{s.t. }~
& \frac{\left\|\mathbf{c}_m\mathbf{H}_{k,m}\mathbf{W}_k\mathbf{s}_k\right\|^2}
{|\mathcal{S}_m|\sigma_k^2}\ \ge\ \Gamma,\quad \forall m, \label{eq:P1_sdr}\\
& \left| \Im \left\{ \mathbf{c}_m\mathbf{H}_{k,m}\mathbf{W}_k\mathbf{s}_k(s_{k,m})^* \right\} \right|
\leq \tan \theta \notag\\
&\times\left( \Re \left\{ \mathbf{c}_m\mathbf{H}_{k,m}\mathbf{W}_k\mathbf{s}_k(s_{k,m})^* \right\}
- \sqrt{\Gamma\sigma_{\mathrm{eff}}^2} \right),\ \forall m, \label{eq:P1_ci}\\
&  \left\| \mathbf{W}_k \mathbf{s}_k \right\|^2_2 \leq  \frac{p_{\mathrm{tot}}}{L}, \label{eq:P1_pow}\\
&  \| (\mathbf{I}_{N_r} - \mathbf{H}_i \mathbf{H}_i^{\dagger})\mathbf{H}_k\mathbf{W}_k\mathbf{s}_k \|_2^2
\leq \epsilon,\quad \forall i\in \mathcal{A}. \label{eq:P1_anon}
\end{align}
\end{subequations}
Let $\Gamma_{\ell}^{\star}$ denote the optimal value of problem $\mathrm{P1}$ for timeslot $\ell$. The optimal solution  to the original problem  $\mathrm{P0}$ is then  expressed as
\begin{equation}\label{eq:Gamma_overall}
\Gamma^{\star} = \min_{\ell\in\{1,\ldots,L\}} \Gamma_{\ell}^{\star},
\end{equation}
with $\mathbf{W}_k^{\ell}$ selected as the optimizer of the corresponding per-timeslot problem.

Problem $\mathrm{P1}$ is non-convex due to the SDR constraints \eqref{eq:P1_sdr} and the CI constraints
\eqref{eq:P1_ci}. To obtain a tractable design, we adopt\footnote{Although problem~$\mathrm{P1}$ can be tackled via successive convex approximation (SCA) \cite{boyd2004convex,mamaghani2022aerial} by iteratively convexifying the nonconvex SDR constraints, such approaches typically converge only to a stationary point with an initialization-dependent solution and require an additional bisection/inner-loop procedure to handle  $\sqrt{\Gamma}$ in the CI constraints. In contrast, the CI-margin reformulation in~$\mathrm{P2}$ admits an SOCP representation and can be solved globally; hence, we adopt~$\mathrm{P2}$ as a tractable reformulation of~~$\mathrm{P1}$.} a standard CI margin maximization approach \cite{wei2021fundamentals}. Specifically, we define the (real) CI margin  as $\gamma \triangleq \sqrt{\Gamma\sigma_k^2} $, and directly maximize $\gamma$, which enforces a common lower bound on the normalized in-phase constructive component for all
streams.  For QPSK, $|s_{k,m}|=1$, and the CI constraint implies $\Re\{u_m\}\ge \sqrt{|\mathcal{S}_m|}\gamma$, where
$u_m\triangleq \mathbf{c}_m\mathbf{H}_{k,m}\mathbf{W}_k\mathbf{s}_k(s_{k,m})^*$.
Therefore, $|u_m|\ge \Re\{u_m\}\ge \sqrt{|\mathcal{S}_m|}\gamma$, which further yields
$\|\mathbf{c}_m\mathbf{H}_{k,m}\mathbf{W}_k\mathbf{s}_k\|^2 = |u_m|^2 \ge |\mathcal{S}_m|\gamma^2 = \Gamma |\mathcal{S}_m|\sigma_k^2,~\forall m$, i.e., \eqref{eq:P1_sdr} is implied by the CI constraints and can  thus be omitted. Consequently, we equivalently solve
\begin{subequations}\label{eq:P2}
\begin{align}
\mathrm{P2}:~ &\max_{\mathbf{W}_k,\,\gamma}\ \gamma \notag\\
\text{s.t. }~
& \left| \Im \left\{ \mathbf{c}_m\mathbf{H}_{k,m}\mathbf{W}_k\mathbf{s}_k(s_{k,m})^* \right\} \right|
\leq \tan\theta \nonumber\\
& \left(
\Re \left\{ \mathbf{c}_m\mathbf{H}_{k,m}\mathbf{W}_k\mathbf{s}_k(s_{k,m})^* \right\}
-\sqrt{|\mathcal{S}_m|}\gamma \right),\ \forall m, \label{eq:constraint1}\\
&  \left\| \mathbf{W}_k \mathbf{s}_k \right\|^2_2 \leq  \frac{p_{\mathrm{tot}}}{L}, \label{eq:constraint2}\\
&  \| (\mathbf{I}_{N_r} - \mathbf{H}_i \mathbf{H}_i^{\dagger})\mathbf{H}_k\mathbf{W}_k\mathbf{s}_k \|_2^2
\leq \epsilon,\quad \forall i\in \mathcal{A}. \label{eq:constraint3}
\end{align}
\end{subequations}

\begin{algorithm}[!t]
\caption{Anonymous Precoding Algorithm}\label{algo}
\begin{algorithmic}[1]
\item[\textbf{Input:}]
  Channel matrices $\mathbf{H}_1,\dots,\mathbf{H}_K$,
  symbol matrix $\mathbf{S}_k$,
  transmitting user index $k$,
  total transmit power $p_{\mathrm{tot}}$,
  noise variances $\{\sigma_i^2\}_{i=1}^K$,
  number of alias channels $a$, and
  relaxation parameter $\epsilon$.
\State Construct the alias-candidate set $\mathcal{A}^{\mathrm{cand}}_k \!= \!\{1,\ldots,K\}\!\setminus\!\{k\}$.
\State Randomly select $a$ users from $\mathcal{A}^{\mathrm{cand}}_k$ to form the alias set $\mathcal{A}$.
\State Compute the combining matrix $\mathbf{C}$ using \eqref{eq:comb_matrix}.
\For{$\ell=1,\ldots,L$}
\State Construct the CI constraint for $\mathbf{s}_k^\ell$ according to~\eqref{ci_cst_single}.
\State Construct the anonymity constraint using the relaxation parameter $\epsilon$ according to~\eqref{eq:cst_anonym}.
\State Solve optimization problem $\mathrm{P2}$ to obtain $\mathbf{W}_k^\ell$.
\EndFor
\item[\textbf{Output:}]
  Precoding matrices $\{\mathbf{W}_k^\ell\}_{\ell=1}^{L}$ and the
  combining matrix $\mathbf{C}$.
\end{algorithmic}
\end{algorithm}

Problem $\mathrm{P2}$ is a convex second-order cone program (SOCP), due to an affine objective function with linear and second-order cone (SOC) constraints \cite{boyd2004convex}.  The CI constraint in \eqref{eq:constraint1} is
linear, since the absolute-value inequality can be rewritten as two affine
inequalities. The power constraint in \eqref{eq:constraint2} and the anonymity constraint in \eqref{eq:constraint3} are SOC constraints of the form $\|\mathbf{A}x\|_2 \le b$. Hence, $\mathrm{P2}$  can be efficiently solved by any convex optimization tool. Algorithm~\ref{algo} summarizes our proposed anonymous precoding design.

The above SOCP formulation is applicable to any pre-selected alias set $\mathcal{A}$. A simple implementation is to randomly select the alias users from $\mathcal{A}_{k}^{\mathrm{cand}}$, as summarized in Algorithm~\ref{algo}. However, under the proposed heterogeneous-noise-variance identification model, alias users are no longer statistically equivalent, since the likelihood comparison in \eqref{EDi} depends on the noise-variance-dependent term and the variance-weighted signal-dependent terms. Therefore, we further provide a heterogeneous-noise-aware alias selection strategy as an implementation refinement.

The main idea is to first exclude alias sets that are geometrically unfavorable for the SOC anonymity constraints, and then select the final alias set according to which component is expected to have a stronger influence on the receiver-side likelihood comparison. To implement the geometry screening step, for each alias candidate $i$, we define the residual projection matrix as $\mathbf{G}_i=\mathbf{I}_{N_r}-\mathbf{H}_i\mathbf{H}_i^\dagger $. This matrix depends only on the alias-user channel and can therefore be computed before the precoder optimization. Accordingly, an alias set is regarded as geometrically favorable if the true-user channel subspace has a small projection onto the corresponding alias residual subspaces.

Among the retained geometry-favorable alias sets, we construct two candidates based on $E_{\text{const}}$ and the alternative-channel-dependent mismatch term in \eqref{EDi}. The constant-gap candidate, denoted by $\mathcal{A}_c$, favors the noise-variance-dependent component by selecting the set with a larger average $E_{\text{const}}$. The signal-mismatch candidate, denoted by $\mathcal{A}_{s}$, favors a smaller variance-normalized mismatch penalty by selecting the set with a smaller average $\epsilon/\sigma_i^2$.

Finally, we decide which candidate to use according to which component is expected to have a stronger influence on the receiver-side hypothesis separability. Specifically, for the constant-gap candidate $\mathcal{A}^*_c$, we compare its average $E_{\text{const}}$ with an averaged reference signal-mismatch penalty $\frac{1}{a}\sum_{i\in\mathcal{A}_{\mathrm{c}}^*}U_{\mathrm{ref},i}/\sigma_i^2$, where $U_{\mathrm{ref},i}\triangleq \frac{p_{\mathrm{tot}}}{L} \frac{\operatorname{tr}\left(\mathbf H_k^H \mathbf G_i \mathbf H_k\right)} {N_t\|\mathbf H_k\|_F^2}$ denotes a normalized projected-mismatch reference induced by the per-symbol transmit-power budget. If the former is larger, $\mathcal{A}^*_c$ is selected; otherwise, $\mathcal{A}^*_s$ is selected. The resulting heterogeneous-noise-aware alias selection strategy is summarized in Algorithm~\ref{algo2}.

\begin{algorithm}[!t]
\caption{Heterogeneous-Noise-Aware Alias Selection}\label{algo2}
\begin{algorithmic}[1]
\item[\textbf{Input:}]
  Channel  matrices $\mathbf{H}_1,\dots,\mathbf{H}_K$,
  transmitting user index $k$,
  total transmit power $p_{\mathrm{tot}}$,
  noise variances $\{\sigma_i^2\}_{i=1}^K$,
  block length $L$,
  number of data streams $N_s$,
  number of alias channels $a$, and
  relaxation parameter $\epsilon$.
\State Construct the alias-candidate set $\mathcal{A}^{\mathrm{cand}}_k\! =\! \{1,\ldots,K\}\!\setminus\!\{k\}$.
\State For each $i\in\mathcal{A}^{\mathrm{cand}}_k$, compute the residual projection matrix $\mathbf{G}_i$ and the noise-variance-dependent term $E_{\text{const}}$ in \eqref{EDi}.
\State Generate candidate alias sets of size $a$ from $\mathcal{A}^{\mathrm{cand}}_k$.
\State Retain the candidate sets that are geometrically favorable for the SOC anonymity constraints.
\State From the retained sets, obtain the constant-gap candidate $\mathcal{A}_c^*$ by favoring a larger average $E_{\text{const}}$.
\State From the retained sets, obtain the signal-mismatch candidate $\mathcal{A}_{s}^*$ by favoring a smaller average $\epsilon/\sigma_i^2$.
\If{the average $E_{\text{const}}$ over $\mathcal{A}_{\mathrm{c}}^*$ is larger than its average $U_{\mathrm{ref},i}/\sigma_i^2$}
\State Select $\mathcal{A}_c^*$ as the alias set $\mathcal{A}$.
\Else
\State Select $\mathcal{A}_s^*$ as the alias set $\mathcal{A}$.
\EndIf
\item[\textbf{Output:}]
Selected alias set $\mathcal{A}$.
\end{algorithmic}
\end{algorithm}

\subsection{Complexity  Analysis}

The per-timeslot problem~\eqref{eq:P2} is an SOCP and can be solved using a primal--dual interior-point method (IPM).
The worst-case complexity  of conic programs with SOC and  linear matrix inequality (LMI) constraints using a standard IPM scales as $ \mathcal{C}
=\mathcal{O}\!\left(\ln\!\frac{1}{\tau}\,\sqrt{\beta }\, \big(C_{\rm form}+C_{\rm fact}\big)\right)$ \cite{wang2014outage,wei2021fundamentals}, where $\tau>0$ denotes the target accuracy and $\beta$ is the barrier parameter, given by $ \beta=\sum_{j=1}^{p}q_j+2v$, with $p$ denoting the number of LMI constraints of sizes $\{q_j\}_{j=1}^{p}$, and $v$ denoting the number of SOC constraints. The per-iteration computational costs for forming and factorizing the Newton system are upper bounded by $ C_{\rm form}=n_{\rm var}\sum_{j=1}^{p}q_j^{3}+n_{\rm var}^{2}\sum_{j=1}^{p} q_j^{2}+n_{\rm var}\sum_{j=1}^{v} u_j^{2}, C_{\rm fact}=n_{\rm var}^{3} $ \cite{wang2014outage}, where $n_{\rm var}$ is the number of decision variables and $u_j$ is the dimension of the \order{j} SOC constraint.

For problem $\mathrm{P2}$, constraint \eqref{eq:constraint1} is equivalent to two affine inequalities per data stream, yielding $2N_s$ scalar linear constraints. These can be treated as $2N_s$ size-one LMIs, i.e., $p=2N_s$ and $q_j=1$. Constraint \eqref{eq:constraint2} contributes one SOC constraint of size $N_t$, and constraint \eqref{eq:constraint3} contributes $|\mathcal{A}| = a$ SOC constraints of size $N_r$. Hence, $v=1+a$ and $\beta = 2N_s + 2(1+a)$. In addition, since $\mathbf{W}_k\in\mathbb{C}^{N_t\times N_s}$ and $\gamma\in\mathbb{R}$, the real-valued problem has
$n_{\rm var}=2N_tN_s+1 $ decision variables. Substituting these quantities yields the worst-case complexity of solving~$\mathrm{P2}$ per timeslot as $\mathcal{C}_{\textnormal{P2}}
= \mathcal{O}\Big(
\ln\!\frac{1}{\tau}\,
\sqrt{\,2N_s+2(1+a)\,}  \times \left[2N_s\,n_{\rm var}^{2}
+n_{\rm var}\!\left(2N_s+ N_t^{2}+aN_r^{2}\right)
+n_{\rm var}^{3}
\right]\Big) $.

Since the precoder is designed independently across $L$ timeslots, the overall algorithm complexity scales linearly with $L$. Moreover,  since the Newton-system factorization term $n_{\rm var}^{3}$ is dominant in typical MIMO regimes, the total complexity can be approximated as $ \mathcal{C}_{\textnormal{total}}
\approx
\mathcal{O}\!\left(
L\,\ln\!\frac{1}{\tau}\,
\sqrt{2N_s+2(1+a)}\,
(N_tN_s)^3\right).$

We further note that the above complexity scaling remains valid for standard higher-order constellations after redefining the corresponding constructive regions\footnote{ For $M$-PSK, the angular-sector CI constraint can be extended by replacing $\theta=\pi/4$ with $\theta=\pi/M$. For $M$-QAM, inner, edge, and corner constellation points lead to symbol-dependent affine equality and/or one-sided affine inequality constraints on the real and imaginary components of the noiseless post-combined symbol. For $M$-PAM, the one-dimensional constellation structure leads to scalar affine amplitude constraints, where inner and outer symbols correspond to amplitude-equality and one-sided inequality constraints, respectively. In these cases, redefining the constructive regions changes only the affine CI constraints, while the decision-variable dimension, the number of the SOC constraints, and the dominant complexity order of problem~$\mathrm{P2}$ remain unchanged.}.

\section{Simulation Results}\label{sec:simulations}
In this section, we present simulation results to evaluate the detection performance and anonymity behavior of the proposed framework. In each Monte Carlo trial, the index of the active user~$k$ is randomly selected from all~$K$ candidates to ensure statistical fairness. The alias user is then randomly chosen from the remaining $K-1$ candidates to form an alias set with size $a=1$. Unless otherwise specified, the random alias selection method in Algorithm~\ref{algo} is used as the default implementation. The heterogeneous-noise-aware alias selection method in Algorithm~\ref{algo2} is evaluated separately in the simulations where it is explicitly specified, using the same system parameters for a fair comparison. The following  parameters are used throughout the simulations: The total number of users is $K=15$, the number of receive antennas is $N_r = 12$, the number of transmit antennas is $N_t = 8$, the number of data streams is $N_s = 4$, and the block length is $L = 30$.  In addition, the noise variance at any non-transmitting user $i \neq k$ is modeled as a random perturbation around that of the true transmitter, given by $ \sigma_i^2 = 10^{\beta_i/10}\, \sigma_k^2, \forall i \neq k$, where $\{10^{\beta_i/10}\}$ are i.i.d. random variables uniformly distributed over $[10^{-d/10},\,10^{d/10}]$,  representing the relative deviation of each user's noise variance from that of the true transmitter, and  $d$ denotes the half-range of noise variation in the logarithmic domain. We set $d = 1.5$ dB, which introduces moderate randomness yet controlled heterogeneity  in the noise levels. Finally, we clarify that all SNRs in this section are defined as transmit SNRs. Consistent with Sections~\ref{sec:system}--\ref{sec:anony_design}, symbol-level precoding is performed independently across the $L$ symbol intervals to form $\mathbf{X}_k$, followed by blockwise transmitter identification using the decoded $\mathbf{S}_k$ under each candidate hypothesis.

\begin{remark}[Scope and applicability of performance comparisons]
To compare against related methods, we consider both conventional anonymity-agnostic designs, i.e., singular-value-decomposition (SVD) and minimum-mean-squared-error (MMSE) transceivers, and existing physical-layer anonymous precoder designs in the literature. The latter can be categorized into interference-suppression-based and constructive-interference-based designs as further discussed below.

The interference-suppression-based anonymous (ISA) precoder in \cite{wei2021fundamentals} is developed specifically for the strong-transmitter configuration with $N_r\leq N_t$. In this regime, the full row rank of $\mathbf H_i$ implies $\mathbf I_{N_r}-\mathbf H_i\mathbf H_i^\dagger=\mathbf 0$, rendering the projection-based detector ineffective. The cited work therefore adopts the M-Norm detector and designs the ISA anonymity constraint against it. In our strong-receiver configuration, however, $\mathbf I_{N_r}-\mathbf H_i\mathbf H_i^\dagger\neq\mathbf 0$, and the relevant identification statistic depends on hypothesis-dependent projection residuals, which this ISA constraint does not control.  Moreover, the ISA $N_r$-stream interference-suppression structure is infeasible when $N_r>N_t$, since $\operatorname{rank}(\mathbf H_k\mathbf W_k)\leq\operatorname{rank}(\mathbf H_k)\leq N_t<N_r$, and hence the transmitter cannot independently control all $N_r$ receive-antenna branches as required by the original ISA formulation. Accordingly, the ISA precoder is not applicable to the strong-receiver configuration and transmitter-identification model considered in this work.

Meanwhile, existing constructive-interference-based and related alias-based anonymous frameworks \cite{wei2021fundamentals,wei2022physical1,wei2023phy,cui2025physical,wang2025closed} employ different transmitter-identification models and receiver architectures. These frameworks are based on symbol-unaware identification models with homogeneous noise variances, where the pairwise likelihood or residual-distance comparison depends primarily on the signal-dependent mismatch term. Reducing this mismatch can make the true and alias hypotheses equally likely, or sufficiently close. This equal-likelihood alias interpretation supports the corresponding DER-tunable rules and alias-channel-based combiner constructions. In contrast, our symbol-aware transmitter-identification model with user-dependent noise variances introduces additional noise-variance-dependent terms into the likelihood statistic. Reducing the mismatch alone cannot generally equalize the likelihoods of the true-transmitter and alias-transmitter hypotheses or force their pairwise confusion probability toward $0.5$. A numerical comparison would therefore require modifying fundamental components of the prior schemes, including the receiver-side identification model and the receiver-combiner construction.  Such adaptations would no longer represent the original schemes under their intended assumptions and would not constitute a fair comparison.

Therefore, the existing physical-layer anonymous transceiver designs do not constitute feasible and fair performance benchmarks. Hence, we include numerical results of the SVD and MMSE transceivers for performance comparison.
\end{remark}

\begin{figure}[!t]
	\centering
	\includegraphics[width=0.98\linewidth]{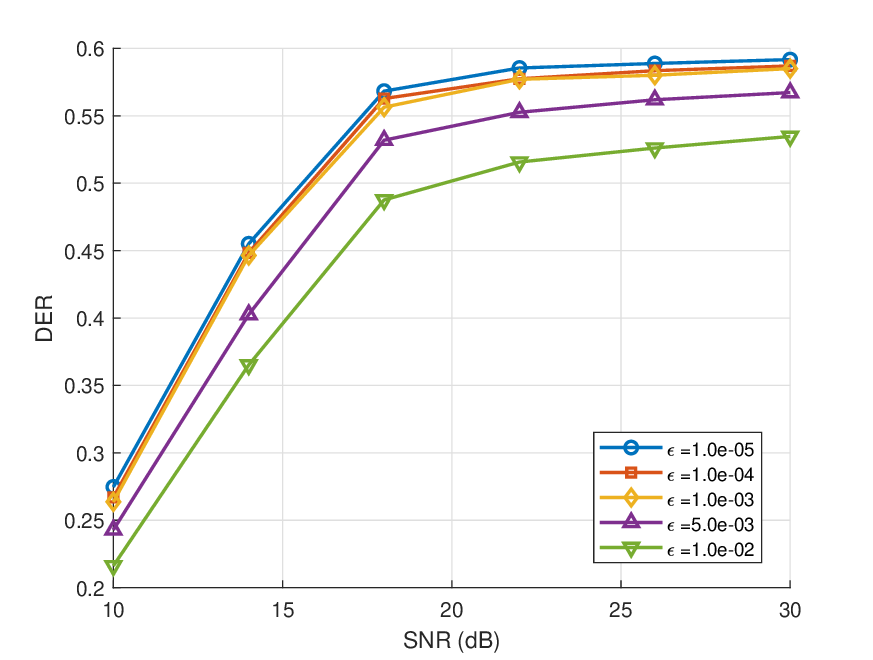}
	\caption{DER vs. SNR under different anonymity relaxations $\epsilon$.}
	\label{DERvsSNR}
\end{figure}

Fig.~\ref{DERvsSNR} shows the DER performance as a function of the SNR for different anonymity constraint relaxation parameters $\epsilon$. Interestingly, the DER increases with the SNR,  indicating that higher transmit power enhances anonymity. This behavior can be explained by examining the expected difference expression $\operatorname{E}[D_i]$ in Lemma \ref{lemma_kld}. Specifically, the noise-variance-related constant term remains essentially independent of the SNR, whereas the signal-dependent terms become more influential as the signal power increases. Since the proposed anonymous precoder explicitly constrains the alternative-channel-dependent component through \eqref{eq:cst_anonym}, the effect of this constraint on the hypothesis comparison becomes more pronounced at high SNRs, i.e., degrading the receiver's identification capability and increasing the   DER.
Moreover, smaller $\epsilon$ generally imposes tighter anonymity constraints and consistently yields higher DER. However, when $\epsilon$ becomes sufficiently small (e.g., below $10^{-3}$), the curves exhibit saturation, suggesting that further tightening provides negligible anonymity improvement.

\begin{figure}[!t]
	\centering
	\includegraphics[width=0.98\linewidth]{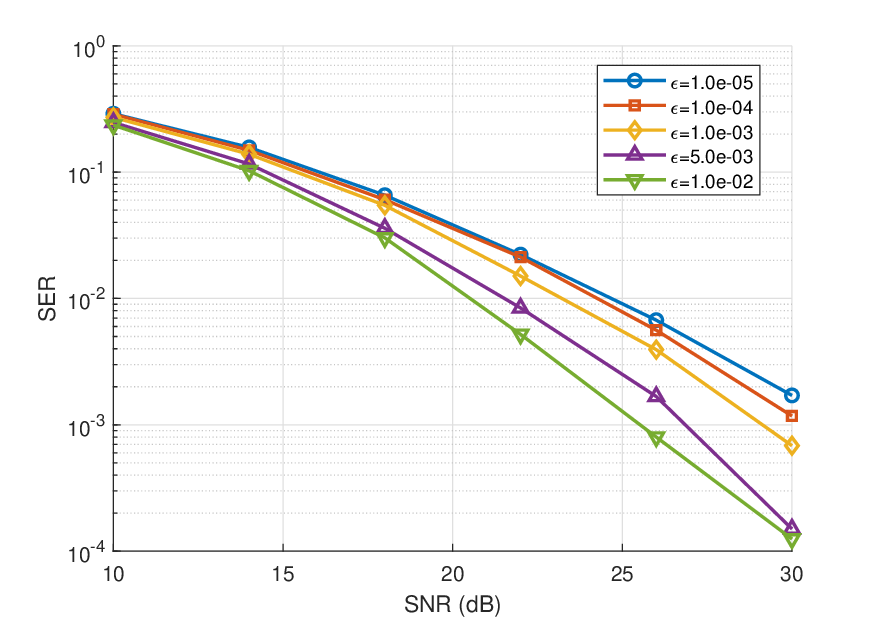}
	\caption{Average symbol error rate (SER) vs. SNR with different $\epsilon$.}
	\label{SERSNR}
\end{figure}

Fig. \ref{SERSNR} illustrates the SER performance against the SNR for different anonymity relaxations. As expected, increasing the SNR improves communication reliability, resulting in SER reduction. Moreover, comparing curves with different $\epsilon$ values reveals that looser anonymity constraints (larger $\epsilon$) achieve lower SER, since the anonymity precoder has greater freedom to optimize signal alignment for detection. Conversely, tighter anonymity enforcement restricts the feasible precoding space and slightly degrades reliability. These results clearly demonstrate the inherent tradeoff between anonymity and communication performance. When the anonymity constraint becomes sufficiently tight (e.g., $10^{-5}$ or even $10^{-4}$), further reduction in $\epsilon$ brings negligible changes to the overall performance, indicating that reliability degradation also exhibits a saturation effect similar to results in Fig.~\ref{DERvsSNR}.

\begin{figure}[!t]
    \centering
    \includegraphics[width=1\linewidth]{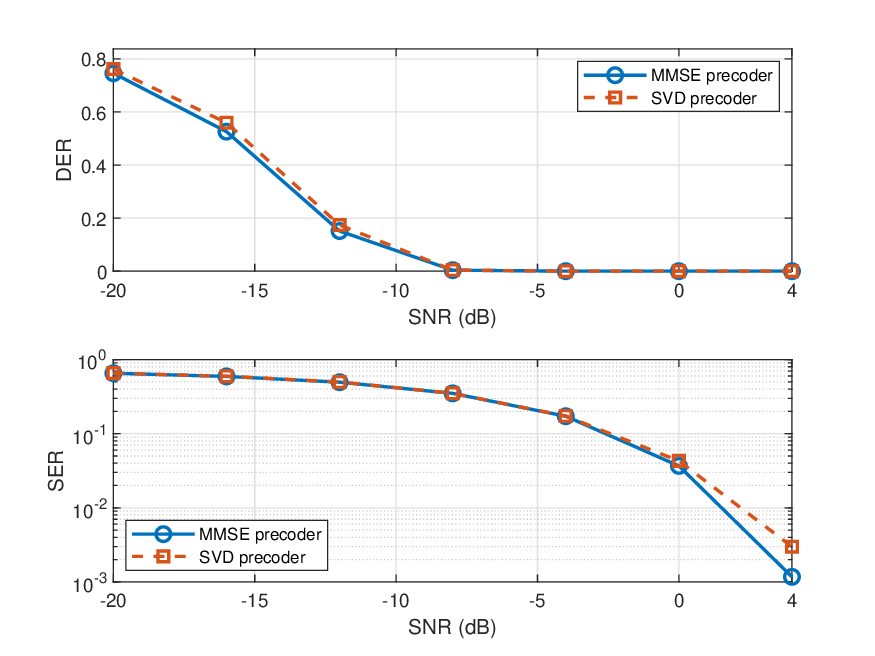}
    \caption{DER and SER vs. SNR of the conventional MMSE and SVD transceivers. The upper and lower subfigures show the DER and SER, respectively.}
    \label{fig:mmse_svd}
\end{figure}

We next compare the proposed design with the conventional MMSE and SVD transceivers. For these conventional benchmarks, the DER is evaluated using the same symbol-aware GLRT, whereas the SER is evaluated using their respective standard channel-dependent combiners to provide a favorable indication of their communication performance. Collectively, Figs.~\ref{DERvsSNR}--\ref{fig:mmse_svd} further illustrate the inherent tradeoff between transmitter anonymity and communication reliability. At very low SNRs, the conventional transceivers exhibit appreciable DER only when their SER remains prohibitively high, reflecting the simultaneous failure of reliable decoding and transmitter identification rather than a practically useful anonymous operating point. Once reliable decoding becomes achievable, their DERs rapidly approach zero. In contrast, the proposed transceiver deliberately reduces the statistical distinguishability among the transmitter hypotheses, achieving a higher DER in a regime where reliable decoding remains possible.

\begin{figure*}[!t]
\centering
\begin{minipage}[t]{0.46\textwidth}
    \centering
    \includegraphics[width=\linewidth]{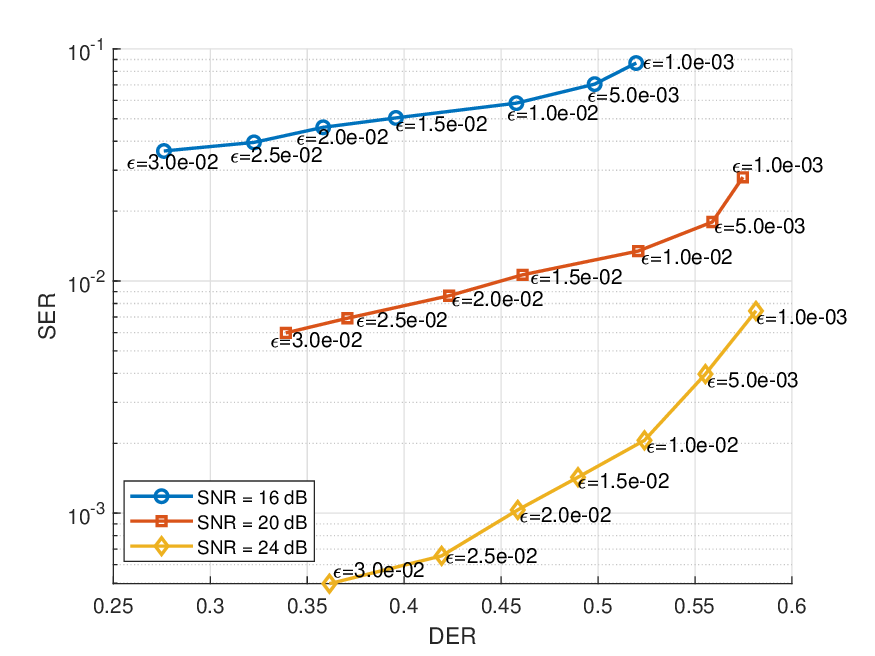}
    \vspace{2pt}
    \centerline{\footnotesize (a) Random alias selection.}
\end{minipage}
\hfill
\begin{minipage}[t]{0.46\textwidth}
    \centering
    \includegraphics[width=\linewidth]{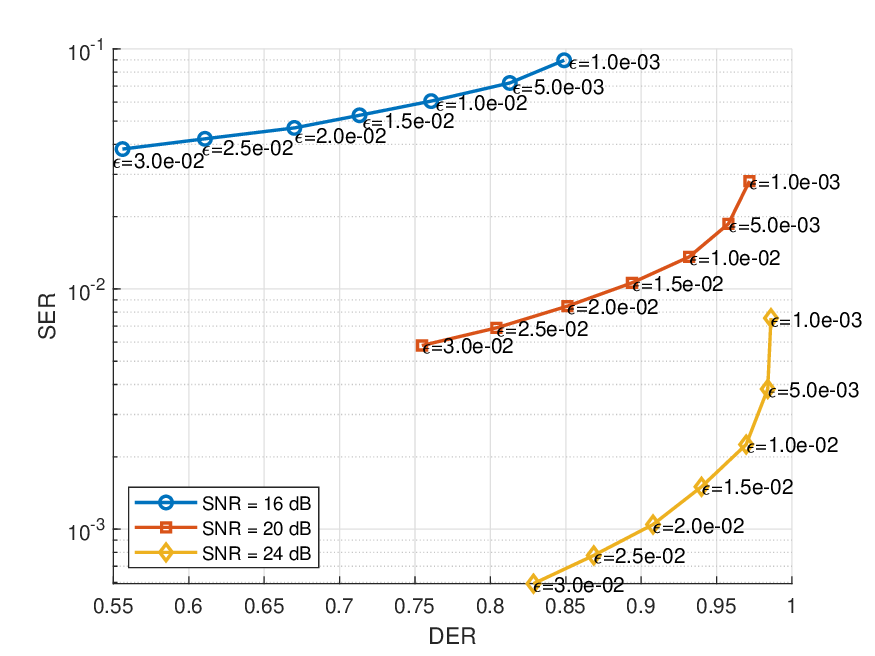}
   \vspace{2pt}
    \centerline{\footnotesize (b) Heterogeneous-noise-aware alias selection.}
\end{minipage}
\caption{SER-DER tradeoff comparison between different alias selection methods under different transmit SNRs and $\epsilon$.}
\label{derser0}
\end{figure*}

Fig.~\ref{derser0} illustrates the SER-DER relationship  under different transmit SNRs and anonymity relaxation parameters $\epsilon$, where Fig.~\ref{derser0}(a) corresponds to the random alias selection strategy and Fig.~\ref{derser0}(b) corresponds to the proposed heterogeneous-noise-aware alias selection strategy. In both cases, each point on the curves corresponds to a different value of $\epsilon$. As $\epsilon$ decreases, i.e., the anonymity constraint becomes tighter (higher DER), the SER increases, implying degraded communication reliability. This behavior explicitly demonstrates the cost of anonymity in terms of communication reliability. For all SNR levels, the curves display a consistent monotonic tradeoff pattern between anonymity and reliability.

By comparing Fig.~\ref{derser0}(a) and Fig.~\ref{derser0}(b), we observe that the proposed heterogeneous-noise-aware alias selection method significantly improves the SER-DER tradeoff compared with random alias selection. Specifically, under the same transmit SNR and a comparable SER level, the alias-selection-based design achieves a much higher DER, indicating stronger transmitter anonymity. This improvement confirms that, under the proposed heterogeneous-noise-variance identification model, alias users are not statistically equivalent, and exploiting the noise-variance-dependent term in the alias selection process can substantially improve anonymity performance. In addition, for both alias selection methods, higher SNR generally leads to improved reliability at a given anonymity level. In addition, when $\epsilon$ becomes very small, the SER may increase sharply while the DER improvement becomes marginal, indicating diminishing anonymity returns in the tight-anonymity regime.

Fig.~\ref{DERvsNs20dB} investigates the impact of the number of transmitted data streams $N_s$ on the anonymity-reliability tradeoff, where the SNR is fixed at $20$~dB. Fig.~\ref{DERvsNs20dB}(a) corresponds to random alias selection, while Fig.~\ref{DERvsNs20dB}(b) corresponds to the proposed heterogeneous-noise-aware alias selection. In both cases, decreasing $\epsilon$ imposes a tighter anonymity constraint, which generally increases the DER at the cost of a higher SER. We also observe that increasing $N_s$ generally improves the DER, but leads to a higher SER. This is consistent with intuition: increasing the number of data streams imposes a more stringent spatial-multiplexing requirement on the anonymous transceiver design, thereby increasing the SER. Meanwhile, due to the fundamental anonymity-reliability tradeoff, the DER generally increases as well. Compared with random alias selection, the heterogeneous-noise-aware alias selection significantly shifts the tradeoff curves toward higher DER values for all considered $N_s$, indicating that exploiting transmitter-dependent noise variances in alias selection is beneficial across different multiplexing configurations.

\begin{figure*}[t]
\centering
\begin{minipage}[t]{0.46\textwidth}
    \centering
    \includegraphics[width=\linewidth]{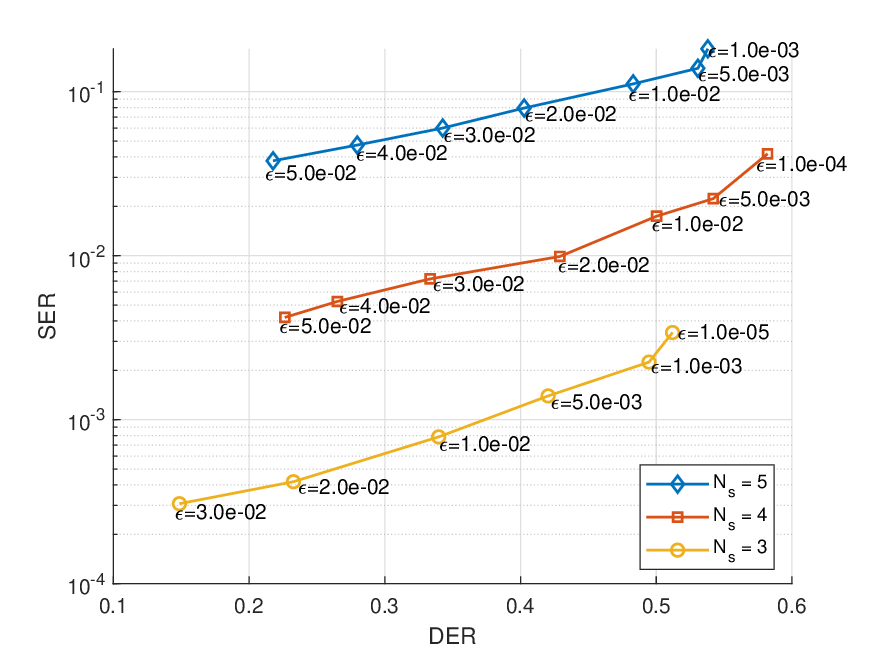}
    \vspace{2pt}
    \centerline{\footnotesize (a) Random alias selection.}
\end{minipage}
\hfill
\begin{minipage}[t]{0.46\textwidth}
    \centering
    \includegraphics[width=\linewidth]{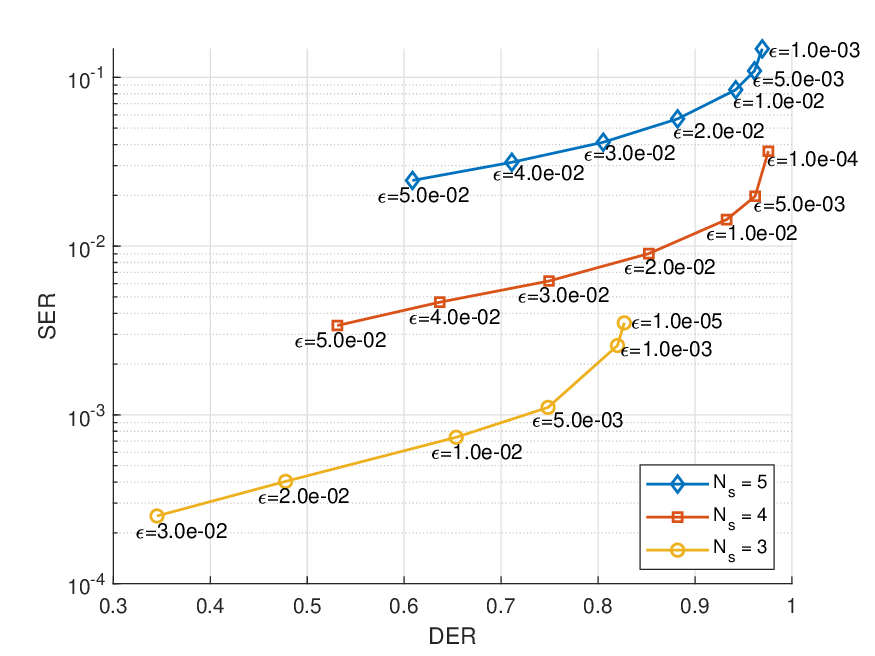}
   \vspace{2pt}
    \centerline{\footnotesize (b) Heterogeneous-noise-aware alias selection.}
\end{minipage}
\caption{SER-DER tradeoff under different numbers of data streams $N_s$ and anonymity relaxation parameter $\epsilon$, with random alias selection and heterogeneous-noise-aware alias selection. The SNR is set to $20$~dB.}
\label{DERvsNs20dB}
\end{figure*}

In Fig. \ref{fig:DER_vs_d}, we investigate how the heterogeneity of user-dependent noise levels affects the system’s anonymity performance. To this end, we plot the DER as a function of the log-domain noise-variation half-range parameter $d$.
At the low SNR (Fig.~\ref{fig:DER_vs_d}(a), $10$~dB), the DER decreases as $d$ increases, indicating reduced anonymity. In this regime, transmitter identification is primarily determined by noise statistics, and larger dispersion in user noise variances makes the competing hypotheses more distinguishable, which enables more reliable detection. In contrast, at the high SNR (Fig.~\ref{fig:DER_vs_d}(b), $20$~dB), the trend reverses and the DER increases with $d$. Here, the signal-dependent terms in \eqref{EDi} become more influential in the hypothesis comparison. For an alias with $\sigma_i^2>\sigma_k^2$, its larger noise variance makes the coefficient of the common blockwise fitting residual $\eta_k$ positive and makes the mismatch penalty $-\Delta/\sigma_i^2$ less negative. Both effects increase $\operatorname{E}[D_i]$ and thereby favor a higher DER. These results demonstrate that the effect of hardware-induced noise heterogeneity on transmitter identification is strongly dependent on the operating SNR.

\begin{figure*}[!t]
\centering
\begin{minipage}[t]{0.47\textwidth}
    \centering
    \includegraphics[width=\linewidth]{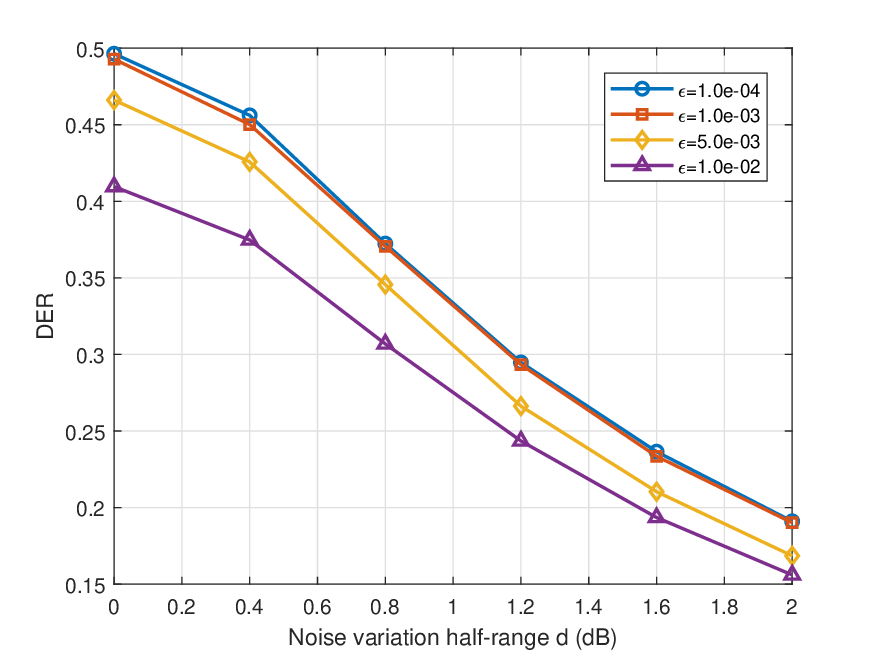}
    \vspace{2pt}
    \centerline{\footnotesize (a) SNR = 10 dB}
\end{minipage}
\hfill
\begin{minipage}[t]{0.47\textwidth}
    \centering
    \includegraphics[width=\linewidth]{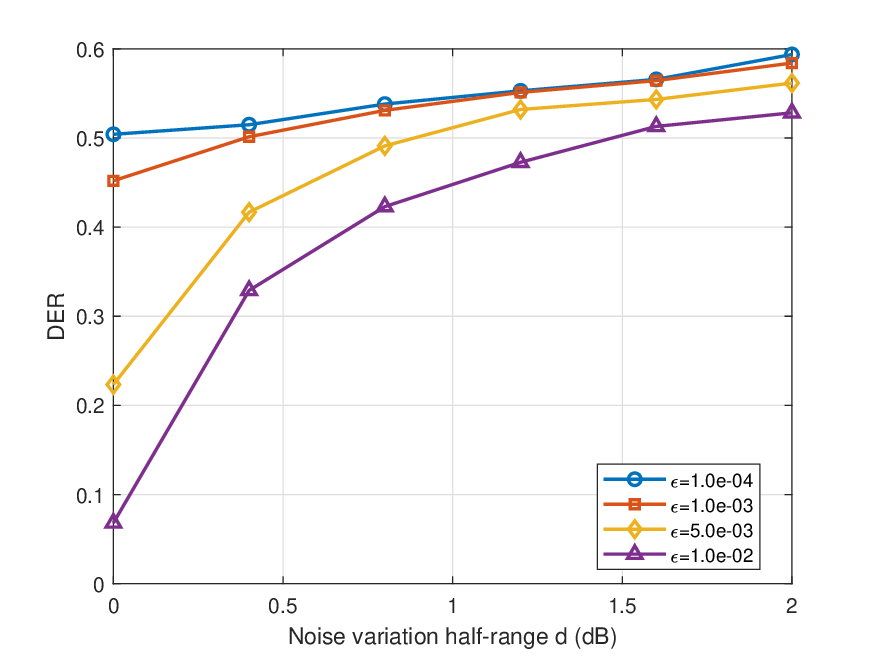}
   \vspace{2pt}
    \centerline{\footnotesize (b) SNR = 20 dB}
\end{minipage}
\caption{DER versus the noise-variation half-range parameter $d$ under different anonymity relaxation parameter $\epsilon$.}
\label{fig:DER_vs_d}
\end{figure*}

\begin{figure}
	\centering
    \includegraphics[width=0.9\linewidth]{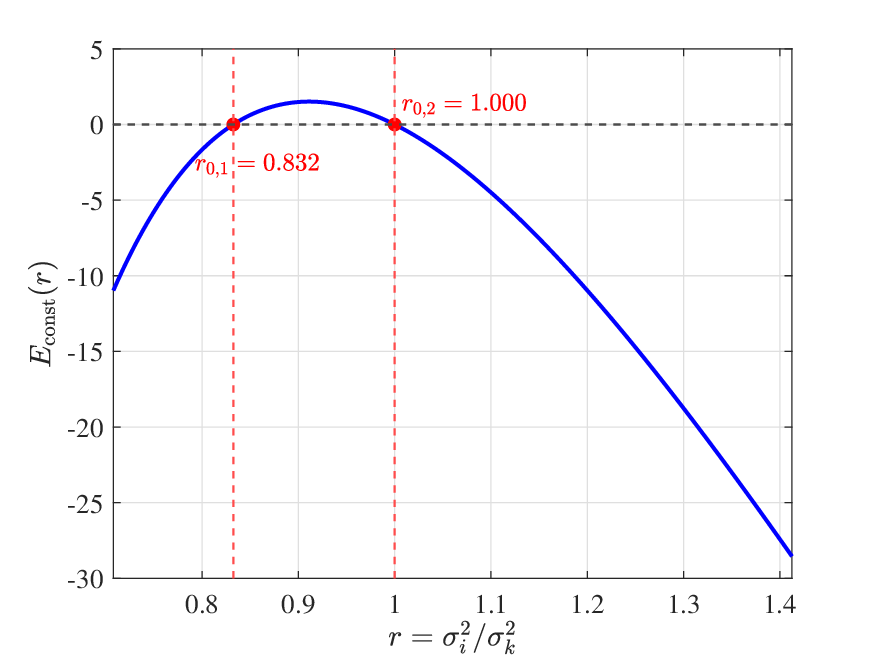}
	\caption{Noise-dependent constant term $E_{\text{const}}$ vs. the noise-variance ratio $r = \sigma_i^2 / \sigma_k^2$.}
	\label{Econst_vs_ratio}
\end{figure}

Fig.~\ref{Econst_vs_ratio} shows the noise-dependent constant term $E_{\text{const}}$ as a function of the noise-variance ratio $r=\sigma_i^2/\sigma_k^2$. This figure is used to isolate the effect of the precoder-independent noise-variance component in the expected log-likelihood difference $\operatorname{E}[D_i]= E_{\text{const}}+\left(\frac{1}{\sigma_k^2}- \frac{1}{\sigma_i^2}\right)\eta_k -\frac{1}{\sigma_i^2}\Delta $. The curve is non-monotonic with respect to \(r\), intersecting zero at \(r=1\) and at a second point \(r_-<1\). Hence, \(E_{\text{const}}\) is positive only within the narrow interval \(r_-<r<1\), namely, when the alias user has a slightly smaller noise variance than the true user. In this region, this noise-dependent component tends to increase the alias test statistic relative to the true-user test statistic. Therefore, when $E_{\text{const}}$ dominates, this region is more favorable for confusing the receiver-side identity detector. Outside this region, this constant component tends to increase the separability between the true and alias hypotheses.

This observation helps explain the DER trends in Fig.~\ref{fig:DER_vs_d}. At low SNRs, the signal-dependent terms are relatively weak, such that the precoder-independent noise component $E_{\text{const}}$ has a stronger influence on the hypothesis comparison. Larger differences in user-dependent noise variances make the hypotheses easier to distinguish. As a result, the DER decreases as the noise-variation range increases, as shown in Fig.~\ref{fig:DER_vs_d}(a). In contrast, at high SNRs, the DER trend is more strongly governed by the variation of the signal-dependent terms. Although the noise-dependent constant component still exists, the increase in the alias noise variance reduces the magnitude of the negative mismatch penalty $-\Delta/\sigma_i^2$ through the factor $1/\sigma_i^2$, while simultaneously increasing the coefficient of the common fitting-residual term $\eta_k$. The resulting reduction in the separability can offset, or even outweigh, the separability increase caused by the noise-statistics difference, leading to a higher DER as the noise heterogeneity increases, as shown in Fig.~\ref{fig:DER_vs_d}(b). Thus, the effect of noise heterogeneity on the DER depends on the relative influence of these components.

\begin{figure}[t]
	\centering
	\includegraphics[width=0.92\linewidth]{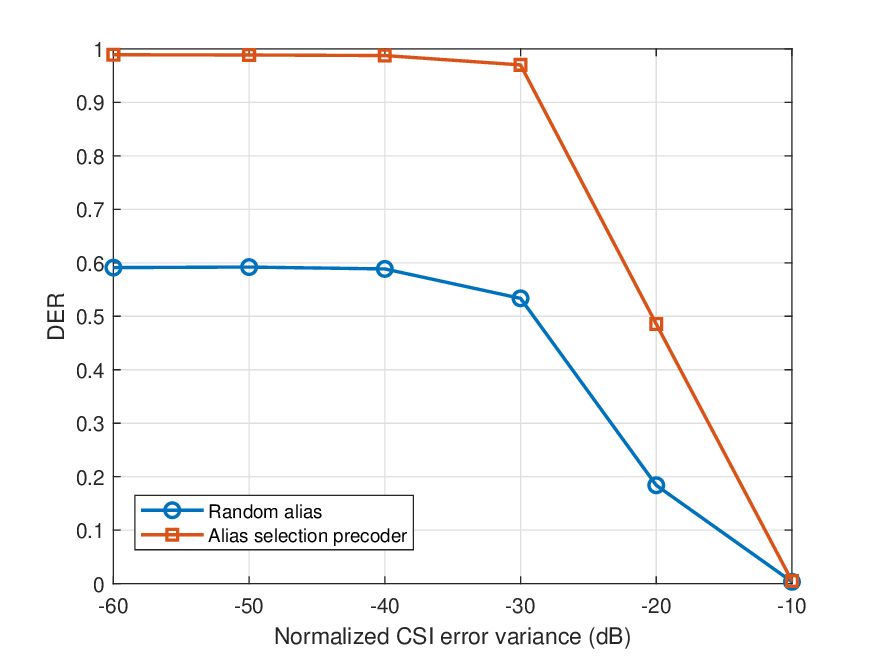}
	\caption{DER performance of the random-alias precoder and the alias-selection-based precoder under imperfect CSI.}
	\label{fig:imperfectCSI}
\end{figure}

Finally, we evaluate the sensitivity of the proposed designs to imperfect CSI. Specifically, we adopt the additive CSI error model, given by $\widehat{\mathbf{H}}_i=\mathbf{H}_i+\Delta\mathbf{H}_i,\quad \forall i\in\mathcal{K} $, where $\mathbf{H}_i$ is the true channel and $\widehat{\mathbf{H}}_i$ is the estimated channel used for precoder design. Each element of the estimation error matrix is independently generated as $ [\Delta \mathbf H_i]_{m,n} \sim \mathcal{CN}(0,\sigma_e^2) $. Since the channel coefficients are generated according to $\mathcal{CN}(0,1)$, $\sigma_e^2$ represents the normalized CSI error variance. We consider the worst-case anonymity-evaluation setting, where the transmitter designs the anonymous precoder based on the imperfect CSI, whereas the receiver is assumed to have perfect CSI for transmitter identification. The SNR is fixed at $25$ dB. We compare the DER performance of the random-alias precoder and the heterogeneous-noise-aware alias-selection precoder under different normalized CSI error variances, as shown in Fig.~\ref{fig:imperfectCSI}. The results show that both schemes maintain favorable DER performance when the CSI error variance is small. The alias-selection-based precoder achieves a higher DER than the random-alias precoder over a wide range of CSI error variances. When the CSI error variance becomes large, the DER of both schemes decreases sharply, because the precoder is designed based on increasingly inaccurate channel information while the receiver still uses the true CSI for transmitter identification.

\section{Conclusion}\label{sec:conclusion}
 We studied physical-layer transmitter anonymity in uplink multi-antenna systems, where we considered a symbol-aware receiver performing multi-hypothesis transmitter identification using the decoded symbols and transmitter-dependent noise variances. We introduced a KLD-based metric to quantify pairwise anonymity between the active and alias users and derived the corresponding anonymity constraint. We further proposed the P-EGC to enhance communication reliability at the receiver and developed the corresponding anonymous precoding scheme under spatial multiplexing. Simulation results, measured in terms of the DER and SER, showed that the proposed scheme significantly improves transmitter anonymity compared with anonymity-unaware precoding schemes while maintaining reliable communication performance. The results further show that the heterogeneous-noise-aware alias selection strategy improves the SER-DER tradeoff compared with random alias selection.

\appendices
\numberwithin{equation}{section}
\makeatletter
\newcommand{\section@cntformat}{Appendix \thesection:\ }
\makeatother

\section{Proof of Lemma \ref{lemma_kld}}\label{Appendix_A}
This appendix derives the KLD-based separability metric between the true hypothesis $\mathcal{H}_k$ and an alternative hypothesis $\mathcal{H}_i$,  as follows:

From \eqref{test_stats}, each test statistic consists of a linear combination of a constant term and a random term involving the Frobenius norm. Hence, we define
\begin{align}\label{F_ki}
    F_k=\| \mathbf{\widehat{Y}}_k - \mathbf{Y} \|_F^2 \text{ and } F_i=\| \mathbf{\widehat{Y}}_i - \mathbf{Y} \|_F^2.
\end{align}
Substituting the reconstructed signals \eqref{Yk_hat} and \eqref{Yi_hat} into \eqref{F_ki} yields
\begin{equation}
\label{Fnorm}
\begin{split}
F_k &= \left\|\mathbf{H}_k\mathbf{X}_k\left( \mathbf{S}_k^\dagger\mathbf{S}_k-\mathbf{I}_L\right) + \mathbf{H}_k\mathbf{H}_k^\dagger \mathbf{N}_k\mathbf{S}_k^\dagger\mathbf{S}_k - \mathbf{N}_k \right\|_F^2,\\
F_i &= \left\| \mathbf{H}_i\mathbf{H}_i^\dagger \mathbf{Y}\mathbf{S}_k^\dagger\mathbf{S}_k - \mathbf{Y} \right\|_F^2.
\end{split}
\end{equation}
For notational convenience, we define signal-dependent quantities as $\eta_k
\triangleq
\left\|\mathbf{H}_k\mathbf{X}_k\left(
\mathbf{S}_k^\dagger\mathbf{S}_k-\mathbf{I}_L\right)
\right\|_F^2$ and $\Delta =\left\|\left(\mathbf{H}_i\mathbf{H}_i^\dagger-\mathbf{I}_{N_r}\right)\mathbf{H}_k\mathbf{X}_k\mathbf{S}_k^\dagger\mathbf{S}_k\right\|_F^2$, where $\eta_k$ represents the energy of the signal component common to the blockwise fitting residuals under $\mathcal{H}_k$ and $\mathcal{H}_i$, and $\Delta$ represents the energy of the additional component associated with the alternative hypothesis $\mathcal{H}_i$.

We first compute the expectations of $F_k$ and $F_i$. Using the properties of the Frobenius norm and the fact that the entries of $\mathbf{N}_k$ are i.i.d. according to $\mathcal{CN}(0,\sigma_k^2)$, the cross term between the deterministic signal term and the noise term in $F_k$ vanishes in expectation. Hence,
the expectation of $F_k$ is given by
\begin{equation}
\label{meanFk}
\begin{split}
\operatorname{E}[F_k]
&=
\eta_k
+
\operatorname{E}\left[
\left\|
\mathbf{H}_k\mathbf{H}_k^\dagger
\mathbf{N}_k\mathbf{S}_k^\dagger\mathbf{S}_k
-
\mathbf{N}_k
\right\|_F^2
\right]\\
&=
\eta_k
+
\operatorname{tr}\left(
\operatorname{E}[\mathbf{N}_k^H\mathbf{N}_k]
\right)
-
\operatorname{tr}\left(
\operatorname{E}\left[
\mathbf{N}_k\mathbf{S}_k^\dagger\mathbf{S}_k
\mathbf{N}_k^H
\right]
\mathbf{H}_k\mathbf{H}_k^\dagger
\right)\\
&=
\eta_k+\sigma_k^2N_rL-\sigma_k^2N_tN_s.
\end{split}
\end{equation}
Similarly, the deterministic signal terms associated with $\eta_k$ and $\Delta$ are Frobenius-orthogonal since $\mathbf{S}_k^\dagger\mathbf{S}_k$ is an orthogonal projection matrix and $\mathbf{S}_k^\dagger\mathbf{S}_k (\mathbf{S}_k^\dagger\mathbf{S}_k- \mathbf{I}_L)=\mathbf{0}$, while their cross terms with the zero-mean noise term vanish in expectation. Therefore, the expectation of $F_i$ is obtained as
\begin{equation}
\label{meanFi}
\begin{split}
\operatorname{E}[F_i] &= \operatorname{E}\left[
\left\|\mathbf{H}_i\mathbf{H}_i^\dagger\mathbf{Y}\mathbf{S}_k^\dagger\mathbf{S}_k - \mathbf{Y} \right\|_F^2 \right]\\
=& \eta_k+\Delta+\operatorname{E}\left[ \left\| \mathbf{H}_i\mathbf{H}_i^\dagger\mathbf{N}_k\mathbf{S}_k^\dagger\mathbf{S}_k-\mathbf{N}_k\right\|_F^2\right]\\
=& \eta_k+\Delta+\operatorname{tr}\left(\operatorname{E}[\mathbf{N}_k^H\mathbf{N}_k]\right)-\operatorname{tr}\left(\operatorname{E}[\mathbf{N}_k\mathbf{S}_k^\dagger\mathbf{S}_k\mathbf{N}_k^H]\mathbf{H}_i\mathbf{H}_i^\dagger\right)\\
=& \eta_k+\Delta+\sigma_k^2N_rL-\sigma_k^2N_tN_s.
\end{split}
\end{equation}
Consequently, the corresponding expectations of test statistics $T_k$ and $T_i$ follow directly by linearity of expectation as
\begin{equation}
\label{Tsmeans}
\begin{split}
\operatorname{E}[T_k]
&=
-N_rL\ln(\pi\sigma_k^2)
-
(N_rL-N_sN_t)
-
\frac{\eta_k}{\sigma_k^2},
\\
\operatorname{E}[T_i]
&=
-N_rL\ln(\pi\sigma_i^2)
-
\frac{\sigma_k^2}{\sigma_i^2}
(N_rL-N_sN_t)
-
\frac{\eta_k+\Delta}{\sigma_i^2}.
\end{split}
\end{equation}
From \eqref{test_stats} and \eqref{KLdef}, we have $T_k=\ln p_k(\mathbf{Y})$ and $T_i=\ln p_i(\mathbf{Y})$. When the \order{k} user is the actual transmitter, $\mathbf{Y}$ follows $q_k(\mathbf{Y})$, and thus $\operatorname{E}[T_k] = \int q_k(\mathbf{Y}) \ln p_k(\mathbf{Y})\,d\mathbf{Y}$ and $ \operatorname{E}[T_i]=\int q_k(\mathbf{Y}) \ln p_i(\mathbf{Y})\,d\mathbf{Y}$. Taking their difference yields
\begin{align}
\label{ETiETk}
\operatorname{E}[T_i]-\operatorname{E}[T_k] &=\int q_k(\mathbf{Y})\left(\ln p_i(\mathbf{Y})-\ln p_k(\mathbf{Y})\right)d\mathbf{Y} \nonumber\\
&= -\int q_k(\mathbf{Y})\ln\frac{p_k(\mathbf{Y})}{p_i(\mathbf{Y})}\,d\mathbf{Y}=-\mathcal{D}_{k,i},
\end{align}
which implies that the KLD-based separability metric defined in \eqref{KLdef} can be computed directly from the difference between the expectations of $T_k$ and $T_i$, without requiring the full distributions of $T_k$ and $T_i$. Finally, substituting \eqref{Tsmeans} into \eqref{ETiETk} establishes the result in Lemma~\ref{lemma_kld}.
\hfill$\blacksquare$

\bibliographystyle{IEEEtran}
\bibliography{references.bib}

\end{document}